\documentclass[11pt]{article}
\usepackage[margin=1in]{geometry}
\usepackage{amsmath}
\usepackage{amssymb}
\usepackage{bm}
\usepackage{booktabs}
\usepackage{graphicx}
\usepackage{pgfplots}
\pgfplotsset{compat=1.18}
\usetikzlibrary{arrows.meta,calc,decorations.pathreplacing}
\usepackage[round,authoryear]{natbib}

\usepackage{amsthm}
\newtheorem{proposition}{Proposition}

\newcommand{\Ipred}{N_{\mathrm{pred}}}
\newcommand{\Imeas}{I_{\mathrm{meas}}}
\newcommand{\sigfair}{\sigma_{\mathrm{fair}}}
\newcommand{\sigCRB}{\sigma^{\mathrm{CRB}}}
\newcommand{\sigAlg}{\sigma^{\mathcal{A}}}

\newcommand{\FCRB}{\mathbf{F}}
\newcommand{\xvec}{\mathbf{x}}
\newcommand{\Amat}{\mathbf{A}}

\begin{document}

\title{Per-Bundle Statistical Limits and Learned-Prior Inversion in
Multiplexed X-ray Imaging, with Application to Temporal CT}
\author{G M Besson\\ ForeVision XRCT, Boulder, WY 82923, USA\\
\texttt{ForeVision\_XRCT@protonmail.com}}
\maketitle

\begin{abstract}
\noindent\textbf{Objective.}
Temporal CT (TCT) fires $N_S = 3$ X-ray sources simultaneously onto a
shared detector, creating a pre-reconstruction inverse problem: each
bundle of $M = 5$ Poisson intensities sums unlabeled photon
contributions from $K = 3$ line integrals. Because the measurement sums
exponentials rather than forming one Beer--Lambert product, $\log$ and
$\sum$ do not commute and the inversion is nonlinear. We quantify the
dose cost this multiplexing imposes and how closely estimators can
approach the resulting limit.

\noindent\textbf{Approach.}
We treat the $5\times3$ bundle as a model problem for multiplexed photon
aggregation, with TCT the motivating instance. Closed-form
Cram\'er--Rao bounds (CRBs) are expressed as dose-inflation factors against an
equal-dose single-source floor, and two estimators --- a structured classical
per-bundle estimator (SNN1) and a physics-motivated residual network ---
are benchmarked on three datasets: i.i.d.\ synthetic, an analytical
phantom, and single-patient bundles.

\noindent\textbf{Main results.}
Aggregation imposes a structural loss: at equal attenuation only
$43\%$ of single-source Fisher information survives for the endpoint
paths and $23\%$ for the middle path, fixing constant CRB inflation
ratios $\sqrt{7/3} \approx 1.53$ and $\sqrt{13/3} \approx 2.08$. SNN1
reaches the endpoint CRBs to within a few percent but degrades on the middle
path under photon starvation. A learned joint prior closes much of this gap and, on
single-patient data, pushes middle-path noise below the equal-dose floor
--- a Bayesian effect (interpolation within one anatomy) from the prior,
not the architecture, and not a generalizable dose gain. A mismatched
prior fails out-of-distribution.

\noindent\textbf{Significance.}
The structure (a Poisson sum of
exponentials) and methodology (Fisher-information efficiencies,
dose-inflation accounting) are not specific to TCT: they characterize
any measurement that aggregates attenuated photons before counting,
transferring with only the incidence matrix
respecified. Whether a learned prior yields a generalizable dose
reduction is the open question a companion paper addresses through a
multi-patient corpus.
\end{abstract}

\noindent{\it Keywords}: computed tomography, temporal CT, multi-source CT, X-ray source multiplexing, pre-reconstruction inverse problem, Cram\'er--Rao bound, Fisher information, dose efficiency, learned prior

\section{Introduction}

\subsection*{Temporal resolution, dose, and multi-source CT}

In cardiac and other dynamic CT, temporal resolution sets a hard
ceiling on image quality: anatomical motion during the acquisition
window produces blurring and streak artifacts that reconstruction
cannot undo \citep{kalender2011}. In a conventional single-source
scanner the temporal resolution is governed by the gantry rotation
time, which is bounded by the mechanical and gravitational limits of
spinning a heavy source--detector assembly. Shortening the acquisition
window to freeze motion also collects fewer photons, so temporal
resolution, image noise, and radiation dose are coupled: improving any
one typically taxes another.

One established way past the rotation-speed limit is to use more than
one source. Dual-source CT places two source--detector pairs at a fixed
angular offset, each with its \emph{own} detector, improving temporal
resolution without changing the measurement model --- every detector
still sees a single source \citep{flohr2006dsct}. Temporal CT (TCT)
takes a different route. A wide-aperture detector is \emph{shared}
among $N_S = 3$ sources that are simultaneously in view at every gantry
position \citep{besson2015medphys}, sampling up to $N_S$ angular
positions per interval at a total photon budget of $N_S N_0$. The price
of sharing the detector is that the three sources' photons arrive
\emph{superimposed and unlabeled}: before any reconstruction can
proceed, the aggregated intensities must be separated back into their
constituent line integrals.

This separation is not a linear unmixing. Each detector reading is a sum
of attenuated source intensities --- a sum of exponentials rather than
the single Beer--Lambert product of standard CT --- so the logarithm and
the summation do not commute, and the inversion is nonlinear and far
more poorly conditioned than the per-ray log-transform of single-source
CT. The central questions of this paper are therefore dosimetric as much
as algorithmic: how much dose efficiency does the multiplexing cost
relative to a conventional scan of equal total dose, how much of that
cost is fixed by the geometry, and how much can a well-designed
estimator recover?

\subsection*{The Temporal CT architecture and the bundle structure}

TCT realizes simultaneous multi-source exposure through a double-drum
geometry: a source array and a wide-aperture detector array rotate on
independent, coaxial gantries \citep{besson2015medphys}. Under the
angular-sampling condition described there, $N_S = 3$ sources are in
view of the detector at every gantry position, so three angular samples
are acquired in the interval a single-source scanner needs for one
(Fig.~\ref{fig:geometry}).

Because the three beams illuminate overlapping regions of the shared
detector, a neighborhood of detector cells records intensities to which
more than one source contributes. The measurements therefore partition
into \emph{bundles}: small groups of detector readings that jointly
constrain a common set of line integrals. Inversion is local to a
bundle --- recover the line integrals from its summed-intensity
readings --- and it is this local problem, repeated across the detector
and across views, that stands between the raw TCT data and any
conventional reconstruction.

For analysis we adopt a bundle of $M = 5$ readings constraining $K = 3$
line integrals $(x_1, x_2, x_3)$, related through an $M\times K$ (here
$5\times3$) binary
incidence matrix $\Amat$ (Section~\ref{sec:forward}); its status as a
representative model rather than a specific realized geometry is
addressed below. The structure is deliberately asymmetric: two readings
are \emph{dedicated}, each measuring a single endpoint path ($x_1$ or
$x_3$) in isolation, while the three central readings mix contributions,
so the middle path $x_2$ is never observed alone --- only in aggregate.
This asymmetry is the origin of the unequal information loss that the
remainder of the paper quantifies.

\begin{figure}[t]   
\centering
\begin{tikzpicture}[font=\small]

\begin{scope}
  \fill[gray!12] (2.5,1.9) ellipse (1.9 and 0.95);          
  \fill[gray!18] (0,0) rectangle (5,0.45);                  
  \foreach \i in {0,1,2,3,4,5}{\draw[gray!55] (\i,0) -- (\i,0.45);}
  \foreach \i/\lab in {0/1,1/2,2/3,3/4,4/5}{
     \node[font=\footnotesize] at (\i+0.5,0.225) {$Y_{\lab}$};}
  \node[anchor=north,font=\footnotesize] at (2.5,-0.08){shared wide-aperture detector};
  \fill[blue!50,opacity=0.16]        (1.5,3.5) -- (0,0.45) -- (3,0.45) -- cycle;
  \fill[red!60,opacity=0.16]         (2.5,3.5) -- (1,0.45) -- (4,0.45) -- cycle;
  \fill[green!55!black,opacity=0.16] (3.5,3.5) -- (2,0.45) -- (5,0.45) -- cycle;
  \foreach \sx/\lab in {1.5/1,2.5/2,3.5/3}{
     \node[circle,fill=black,inner sep=1.7pt,
           label={[font=\footnotesize]above:{$S_{\lab}$}}] at (\sx,3.55) {};}
  \node[font=\footnotesize] at (2.5,4.2){$N_S=3$ sources, simultaneous};
  \draw[decorate,decoration={brace,amplitude=5pt,mirror}] (0,-0.5) -- (5,-0.5)
     node[midway,below=4pt,font=\footnotesize]{one bundle ($M=5$ readings)};
  \node[font=\bfseries] at (-0.35,3.9){(a)};
\end{scope}

\begin{scope}[xshift=6.8cm,yshift=0.9cm]
  \def\cw{0.85}\def\ch{0.55}
  \foreach \r/\j in {0/0, 1/0,1/1, 2/0,2/1,2/2, 3/1,3/2, 4/2}{     
     \fill[blue!22] (\j*\cw,{(4-\r)*\ch}) rectangle (\j*\cw+\cw,{(4-\r)*\ch+\ch});}
  \foreach \r in {0,1,2,3,4}{\foreach \j in {0,1,2}{               
     \draw[gray!55] (\j*\cw,{(4-\r)*\ch}) rectangle (\j*\cw+\cw,{(4-\r)*\ch+\ch});}}
  \foreach \j/\lab in {0/1,1/2,2/3}{                               
     \node[font=\footnotesize] at (\j*\cw+\cw/2,{5*\ch+0.22}) {$x_{\lab}$};}
  \foreach \r/\lab in {0/1,1/2,2/3,3/4,4/5}{                       
     \node[anchor=east,font=\footnotesize] at (-0.12,{(4-\r)*\ch+\ch/2}) {$Y_{\lab}$};}
  \draw[green!45!black,very thick] (0,{4*\ch}) rectangle (3*\cw,{5*\ch});  
  \draw[green!45!black,very thick] (0,0)        rectangle (3*\cw,{1*\ch}); 
  \node[anchor=west,font=\footnotesize,green!45!black] at (3*\cw+0.15,{4.5*\ch}){dedicated};
  \node[anchor=west,font=\footnotesize,green!45!black] at (3*\cw+0.15,{0.5*\ch}){dedicated};
  \node[font=\footnotesize,align=center,anchor=north] at (\cw+\cw/2,-0.18)
     {$x_2$ never\\ measured alone};
  \node[anchor=north,font=\footnotesize] at (1.3,-1.0)
     {$Y_j \sim \mathrm{Poisson}\!\big(\textstyle\sum_k A_{jk}\,N_0\,e^{-x_k}\big)$};
  \node[font=\bfseries] at (-0.7,3.0){(b)};
\end{scope}

\end{tikzpicture}
\caption{The Temporal CT bundle problem. (a)~Schematic of the TCT
principle: $N_S=3$ simultaneously active sources share a wide-aperture
detector, and their overlapping beams make each detector reading a sum
of source contributions, grouping the $M=5$ readings of a local
neighbourhood into a \emph{bundle}. The depiction is illustrative; the
realized double-drum geometry is described in~\cite{besson2015medphys}.
(b)~The $5\times3$ model problem analysed here: the binary incidence
matrix $\Amat$ relates the bundle readings $Y_j$ to the $K=3$ line
integrals $x_k$ via $Y_j\sim\mathrm{Poisson}(\sum_k A_{jk}N_0e^{-x_k})$.
The endpoint paths $x_1,x_3$ each have a dedicated reading ($Y_1,Y_5$),
while the middle path $x_2$ appears only in mixed readings --- the
origin of the asymmetric information loss analysed below.}
\label{fig:geometry}
\end{figure}

\subsection*{The $\boldsymbol{5\times3}$ bundle as a model problem}

The $5\times3$ system studied here does not reproduce any specific TCT
geometry, and we make no claim that it is canonical. It is a
\emph{model problem} that captures
the defining difficulty of the TCT principle --- the inversion of
summed, unlabeled photon intensities, in which $\log$ and $\sum$ do not
commute --- together with its characteristic asymmetry (dedicated
endpoint rows, an absent middle-path row, and the resulting unequal
information loss). The study is deliberately confined to the projection
domain --- it characterizes the per-bundle inversion as a
statistical-limits problem, and image formation is out of scope. Three
claims of decreasing strength follow, and we
separate them deliberately. First, the analysis framework --- the
Fisher-information and Cram\'er--Rao machinery and the dose-inflation
accounting --- is geometry-agnostic by construction: only the system
matrix $\Amat$ specializes to a chosen geometry, so the framework
transfers directly to any realized TCT design --- and, more broadly, to
other multiplexed measurements that aggregate attenuated photon
contributions, where it again enters only through $\Amat$. Second, the
\emph{qualitative} difficulty it reveals --- nonlinear path coupling,
the collapse of middle-path information under photon starvation, and the
consequent dose cost --- is generic to any architecture that sums
overlapping attenuated contributions before counting --- coded-source
imaging and overlapping-projection geometries among them --- and is
therefore representative
rather than incidental. Third, the specific
numerical floors ($\sqrt{7/3}$, $\sqrt{13/3}$, and the dose-inflation
values reported below) are illustrative of the \emph{magnitude} of the
difficulty for this instance and would be recomputed for a chosen
geometry. The model is exercised on realistic inputs: the
patient-derived bundles draw their attenuation statistics from a
clinical CT volume rather than a synthetic distribution, so the
difficulty is stressed on real anatomy. Finally, we state what the model
omits. It treats the estimation problem under a monochromatic,
scatter-free idealization; physical source attribution --- plausibly
through spectral encoding of the sources --- and scatter mitigation lie
outside it and are the subject of subsequent work. Drawing these
boundaries is what makes the representativeness claim usable: the model
characterizes the core estimation difficulty of TCT and the methods that
address it, not a complete engineering solution.

\subsection*{Two sources of performance loss}

Two qualitatively different losses degrade performance, and separating
them clarifies what each method can fix. The \emph{aggregation loss} is
structural: summing photons at the detector destroys
the information that the $N_S K = 9$ latent Poisson contributions would
carry if observed individually, and the Cram\'er--Rao bound fixes how
much survives --- at equal attenuation only $3/7 \approx 43\%$ for
endpoint paths and $3/13 \approx 23\%$ for the middle path
(Section~\ref{sec:crb}). No estimator using only the five bundle
measurements can recover the lost fraction. The \emph{algorithmic
inefficiency} is contingent and reducible: given those five
measurements, a classical algorithm extracts only part of the Fisher
information they carry. SNN1 nearly removes this inefficiency for the
endpoints; the neural network addresses what remains on the middle path
not by undoing the aggregation loss but by supplying a learned joint
prior over $(x_1, x_2, x_3)$ that, on anatomically realistic data, can
dominate the collapsed Fisher information and push the estimate below
the frequentist bound --- a concentrated-prior (Bayesian) effect that,
as we discuss, reflects interpolation within the training anatomy rather
than a generalizable recovery of aggregation-lost information. We report
all performance as the dose-inflation
factor of Eq.~\eqref{eq:dif} --- the squared ratio of an estimator's
standard deviation to the equal-dose single-source floor.

\subsection*{Contributions}

This paper makes four contributions. (i)~A closed-form Fisher-information
and CRB analysis of the Poisson sum-of-exponentials bundle, whose
equal-attenuation factorization yields the exactly constant inflation
ratios $\sqrt{7/3}$ and $\sqrt{13/3}$ and bundle efficiencies
$\eta_k = 3/7$ and $3/13$. Because the architecture enters only through
the incidence matrix, these efficiencies serve as figures of merit for
any $M\times K$ multiplexed-measurement geometry, within TCT and beyond;
a second, four-source geometry is worked explicitly in the supplementary
material. (ii)~SNN1, a structured classical per-bundle
estimator attaining the endpoint CRBs to within a few percent. (iii)~An
evaluation of a physics-motivated residual network across three datasets
of increasing sinogram structure (RND~$\to$~SGS~$\to$~PIS), showing that
a concentrated single-anatomy prior can drive the middle-path noise
below the equal-dose single-source floor by interpolating within that
anatomy's bundle statistics, while a mismatched prior fails
out-of-distribution --- evidence that the sub-floor behaviour is
prior-specific rather than a generalizable dose gain. (iv)~An ablation in which raw measurements alone
and measurements-plus-warm-start converge to indistinguishable accuracy,
indicating that the network learns the nonlinear inverse map implicitly.
Three directions that extend beyond the per-bundle problem --- exploiting
inter-bundle sinogram correlations (Section~\ref{sec:interbundle}),
using the efficiencies $\eta_k$ to optimize bundle geometry, and
co-designing the source-flux profile through bow-tie filtration and
tube-current modulation --- are pursued in a companion paper.

\subsection*{Relation to prior work}

Dual-source CT illuminates from two tubes simultaneously but with
non-overlapping beams and detectors, so each tube--detector pair remains
an independent line-integral measurement with no inversion problem of
the present type \citep{flohr2006dsct}. The TCT architecture
\citep{besson2015medphys} is distinguished by simultaneous, overlapping
source projections that require per-bundle pre-reconstruction inversion.
Neural networks for CT have been applied in the image
\citep{chen2017,jin2017} and sinogram \citep{zhang2018} domains. The
overlapping-beam transmission forward model itself has precedent in
emission-tomography attenuation imaging \citep{yu2000}. Cram\'er--Rao
bounds have characterized task-based image quality in CT before
\citep{pineda2012}; the present work instead applies the CRB in the
projection domain, to the per-bundle inversion itself. To our
knowledge, the present work is the first to treat multi-source
CT bundle inversion as a standalone estimation problem, to benchmark
per-bundle estimators against a closed-form per-bundle CRB, and to bring
a learned joint prior to bear on it.

\section{Forward model and problem formulation}
\label{sec:forward}
\subsection{Multi-source intensity model}

We adopt the monochromatic Beer--Lambert model with Poisson photon
statistics: each source emits $N_0$ photons per measurement interval at
a single effective energy, the detector counts individual arrivals, and
the measurement noise is exactly Poisson --- the standard framework for
analysing statistical limits in CT \citep{barrett1994,fessler1994}.
Polychromatic beam hardening, scatter, and energy-integrating detector
response are deferred to future work, isolating the statistical limits
of multi-source aggregation from spectral effects. Throughout, $N_0$ is
the incident photon count per source per interval (the air-scan
calibration value) and all intensities are photon counts.

In single-source CT each detector reading involves one path, with
expected count $N_0 e^{-x}$ for a line integral $x$. In TCT a detector
cell may receive photons from up to $N_S = 3$ paths at once, and because
photons from different sources are physically indistinguishable the
reading is a \emph{sum} of independent Poisson variables,
$\Imeas[j] = \sum_k Y_{j,k}$ with
$Y_{j,k}\sim\mathrm{Poisson}(N_0 A[j,k] e^{-x_k})$. By the additivity of
Poisson variables, $\Imeas[j]$ is itself Poisson with mean
\begin{equation}
  \Ipred[j](\xvec) = N_0 \sum_{k=1}^{3} A[j,k]\, e^{-x_k}.
  \label{eq:forward}
\end{equation}
Equation~\eqref{eq:forward} is a \emph{sum of exponentials}, not the
exponential of a sum. This is the single structural fact from which the
rest of the paper follows: $\log$ and $\sum$ no longer commute, so the
per-ray log-transform that linearizes single-source CT does not apply.
This structure is not specific to TCT --- it is the signature of any
measurement that aggregates attenuated photon contributions before
counting, such as coded-source imaging and overlapping-projection
geometries. The aggregation loss requires that the summed contributions
not be separable by any recorded degree of freedom; a scheme that
resolves one --- an energy-discriminating detector that separates
sources by spectrum, say --- recovers part of the information and partly
escapes the bound, as noted in Section~\ref{subsec:limitations}.
Indeed, the same
summed-transmission structure has long been recognized in
emission-tomography attenuation imaging: with overlapping beams, one
detector element records photons from different sources that traversed
different paths, and \citet{yu2000} developed a monotonic
maximum-likelihood reconstruction to account for it. The present work
differs in isolating the un-summing as a standalone per-bundle
estimation problem with its own statistical limits, rather than folding
the overlap into a full image reconstruction. The
architecture enters the analysis below only through the incidence matrix
$\Amat$, so the Fisher-information and dose-inflation framework that
follows transfers to those settings with $\Amat$ respecified.

\subsection{System matrix and its structure}

For the $N_S = 3$ bundle the system matrix
$\Amat \in \{0,1\}^{M\times K} = \{0,1\}^{5\times3}$ follows the sliding-window source
arrangement
\begin{equation}
  \Amat =
  \begin{pmatrix}
    1&0&0\\ 1&1&0\\ 1&1&1\\ 0&1&1\\ 0&0&1
  \end{pmatrix},
  \label{eq:sysmat}
\end{equation}
so the predicted intensities are $\Ipred[1] = N_0 e^{-x_1}$,
$\Ipred[2] = N_0(e^{-x_1}+e^{-x_2})$,
$\Ipred[3] = N_0(e^{-x_1}+e^{-x_2}+e^{-x_3})$,
$\Ipred[4] = N_0(e^{-x_2}+e^{-x_3})$, and $\Ipred[5] = N_0 e^{-x_3}$.
Three structural features govern the inversion. First, rows~1 and~5 are
\emph{dedicated}: each depends on a single endpoint, so the one-step
estimate $\hat{x}_1 = -\ln(\Imeas[1]/N_0)$ asymptotically attains the
single-ray CRB $e^{x_1/2}/\sqrt{N_0}$, and likewise for $x_3$. Second,
the middle path $x_2$ has \emph{no dedicated row} --- it appears only in
the mixed rows~2--4, always alongside another path --- and this
asymmetry is the geometric origin of its larger noise penalty
($\sqrt{13/3}$ versus $\sqrt{7/3}$; Section~\ref{sec:crb}). Third, the
system is algebraically benign: rows~1 and~5 isolate $x_1$ and $x_3$ and
any mixed row then supplies $x_2$, so $\mathrm{rank}(\Amat) = 3$ and the
condition number is only $\kappa(\Amat) \approx 3.19$. The
limitation that matters is therefore not linear but statistical:
the Fisher information, not $\Amat$, \emph{vanishes} at high
attenuation, its magnitude falling as $e^{-x}$.

\subsection{Known calibration and the latent counts}

$N_0$ is a \emph{known} parameter, not an unknown: in practice it is
obtained from an air-scan calibration --- every detector cell measures the
unattenuated flux with no object in the beam --- and is therefore known for
all bundles in that view. What cannot be known are the $N_S K = 9$
latent Poisson contributions $\{Y_{j,k}\}$, which the detector sums into
five observable counts. The Fisher information measures exactly how much
information about $\xvec$ survives this compression from nine latent to
five observed counts.

\subsection{Log-domain behaviour: the domination principle}

The difficulty is sharpest in the log domain. For a mixed row the
noiseless log-transform $y_j = -\ln(\Imeas[j]/N_0)$ obeys the
log-sum-exp identity
$y_j = \min_k x_k - \ln(1 + \sum_{k\neq k^*} e^{-(x_k - x_{k^*})})$,
where $k^*$ is the least-attenuated contributing path. When the path
contrast is large the correction vanishes and $y_j \approx \min_k x_k$:
the reading is \emph{dominated by the brightest path}, and the heavily
attenuated co-paths become invisible. This domination is a property of
the measurement, not of any estimator, and it is why the middle path
loses information under photon starvation. We therefore develop the
Cram\'er--Rao analysis directly in the photon-count domain of
Eq.~\eqref{eq:forward}, where the Fisher information stays well defined
and the MLE keeps its $e^{x/2}/\sqrt{N_0}$ scaling at any attenuation.

\subsection{The inverse problem}

Given noisy counts $\{\Imeas[j]\}_{j=1}^{5}$ and known $N_0$, the bundle
inversion problem is to estimate $\xvec = (x_1, x_2, x_3)^\top$ with
$x_k \in [0, x_{\max}]$, seeking unbiased (or asymptotically unbiased)
estimators whose variance approaches $[\FCRB^{-1}(\xvec)]_{kk}$. This
bound is the interior, regular-model Cram\'er--Rao limit; at the box
boundaries $x_k \in \{0, x_{\max}\}$ it need not hold, and we flag the
constrained regime where it arises. The
support ceiling $x_{\max} = 9.2$ is the largest line integral through a
$46\times24$~cm body ellipse at $\mu = 0.20$~cm$^{-1}$; at
$N_0 = 10^5$ this transmits only $\approx 10$ photons --- the
photon-starved regime where the middle-path penalty bites.

\subsection{The dose-inflation factor}

To make the cost of multi-source aggregation interpretable in clinical
terms, we measure every estimator against the noise a conventional
single-source CT scan would achieve \emph{at the same total dose}. Here ``dose''
denotes the incident photon budget $N_S N_0$ under a matched
monochromatic idealization --- an information-theoretic normalization
that is a \emph{proxy} for dosimetric dose, not a substitute for it:
equal photon budget need not mean equal absorbed dose (mGy) once
geometry, beam quality, and the spatial distribution of energy
deposition differ between a three-source and a single-source exposure.
TCT expends $N_S = 3$ source exposures of $N_0$ photons per measurement
interval, a total budget of $3 N_0$. A single-source scan given that
same budget attains, for a path of attenuation $x$, the Cram\'er--Rao
standard deviation
\begin{equation}
  \sigfair(x) \;=\; \frac{e^{x/2}}{\sqrt{3 N_0}},
  \label{eq:sigfair}
\end{equation}
the equal-dose single-source floor. For an estimator $\hat{x}_k$ with
standard deviation $\sigma_{\mathrm{algo}}(x_k)$ we define the
\emph{dose-inflation factor}
\begin{equation}
  \mathrm{DIF}_k \;=\;
    \frac{\sigma_{\mathrm{algo}}^2(x_k)}{\sigfair^2(x_k)},
  \qquad
  \frac{\sigma_{\mathrm{algo}}(x_k)}{\sigfair(x_k)}
    \;=\; \sqrt{\mathrm{DIF}_k}.
  \label{eq:dif}
\end{equation}
This dose reading is exact \emph{at the statistical bound}. The Poisson
information is linear in the photon budget ($\mathbf{F}\propto N_0$;
Section~\ref{sec:crb}), so the CRB and the floor $\sigfair$ both scale
as $1/\sqrt{N_0}$ and $\mathrm{DIF}$ is precisely the factor by which an
equally-tasked single-source scan's dose would have to change to match
it: $\mathrm{DIF} > 1$ is a dose penalty of the multi-source geometry,
$\mathrm{DIF} < 1$ a lower-noise image at equal dose. This holds for any
unbiased estimator, whose variance inherits the $1/N_0$ scaling --- but
\emph{not} for a biased, prior-regularized one. A concentrated prior
decouples variance from photon count, so for the learned network
$\sqrt{\mathrm{DIF}}$ is a precision ratio relative to the floor, not a
realizable dose trade. We therefore report $\sqrt{\mathrm{DIF}}$
throughout, reading it as a dose factor at the bound and, where the
prior dominates, strictly as precision relative to the equal-dose
floor. The per-bin $\sqrt{\mathrm{DIF}}$ values reported below are pooled
over the per-bundle dose distribution: $N_0$ is drawn per bundle from a
tube-current-modulation (TCM) model that raises $N_0$ with bundle-mean
attenuation (Section~\ref{subsec:data}), so each pooled figure is an expectation over that
distribution, not a fixed-$N_0$ quantity, while the exact dose-factor
reading above holds bundle-by-bundle at fixed $N_0$.

\section{Statistical limits: Cram\'er--Rao bounds}
\label{sec:crb}
We now quantify the aggregation loss exactly. We fix the single-ray
reference, derive the bundle Fisher information matrix from the Poisson
likelihood, and invert it in closed form at equal attenuation to obtain
the constant inflation ratios at that symmetric reference point.

\subsection{Single-ray reference and the equal-dose floor}
\label{subsec:single_crb}

The Fisher information a measurement carries about a scalar parameter
$x$ is $\mathcal{I}(x) = \mathbb{E}\!\big[(\partial_x \ell)^2\big]
= -\,\mathbb{E}\big[\partial_x^2 \ell\big]$, with $\ell(x)$ the
log-likelihood; it bounds the variance of any unbiased estimator through
the Cram\'er--Rao inequality
$\mathrm{Var}(\hat{x}) \ge \mathcal{I}(x)^{-1}$. For a single
monochromatic Poisson measurement
$N \sim \mathrm{Poisson}(N_0 e^{-x})$ the information evaluates to
$\mathcal{I}(x) = N_0 e^{-x}$, so this bound is
\begin{equation}
  \sigma^{\mathrm{single}}(x)
    = \frac{1}{\sqrt{N_0 e^{-x}}} = \frac{e^{x/2}}{\sqrt{N_0}}.
  \label{eq:single_CRB}
\end{equation}
This grows as the square root of an exponential: each unit of line
integral cuts the transmitted count by $e$ and amplifies the noise by
$\sqrt{e} \approx 1.65$. Because the bundle fires three sources at $N_0$
photons each, the matched reference is an equal-dose single-source scan
at the same total budget $3N_0$; Fisher information being additive over
independent observations, its floor is
$\sigfair(x) = e^{x/2}/\sqrt{3N_0}$ of Eq.~\eqref{eq:sigfair}. Every
inflation ratio below is referenced to $\sigfair$.

\subsection{The bundle Fisher information matrix}
\label{subsec:FIM}

The five measurements are independent Poisson variables with means
$\Ipred[j](\xvec)$, so the negative log-likelihood is
\begin{equation}
  \mathcal{L}(\xvec) = \sum_{j=1}^{5}
    \bigl[\Ipred[j] - \Imeas[j]\ln\Ipred[j]\bigr] + \mathrm{const}.
  \label{eq:NLL}
\end{equation}
With $\partial\Ipred[j]/\partial x_k = -N_0 A[j,k] e^{-x_k}$, the Poisson
score is
\[
  \partial\mathcal{L}/\partial x_k
  = \sum_j (1 - \Imeas[j]/\Ipred[j])\,\partial\Ipred[j]/\partial x_k;
\]
taking the expectation of its outer product and using
$\mathrm{Var}\,\Imeas[j] = \Ipred[j]$ gives the FIM --- the
multiparameter generalization of $\mathcal{I}(x)$ --- in closed form.

\begin{proposition}[Bundle FIM]
\label{prop:FIM}
The Fisher information matrix entries are
\begin{equation}
  F_{ik} = N_0^2 e^{-x_i} e^{-x_k}
    \sum_{j=1}^{5} \frac{A[j,i]\,A[j,k]}{\Ipred[j](\xvec)}.
  \label{eq:FIM_general}
\end{equation}
\end{proposition}
\begin{proof}
For Poisson data $\mathrm{Var}\,\Imeas[j]=\Ipred[j]$, so
\[
  F_{ik}=\sum_{j}\Ipred[j]^{-1}
    (\partial_{x_i}\Ipred[j])(\partial_{x_k}\Ipred[j]).
\]
With $\Ipred[j]=N_0\sum_k A[j,k]\,e^{-x_k}$ one has
$\partial_{x_i}\Ipred[j]=-N_0\,A[j,i]\,e^{-x_i}$; substituting yields
Eq.~\eqref{eq:FIM_general}.
\end{proof}

\noindent Writing $\alpha_k \triangleq e^{-x_k}$ and
$P_j \triangleq \Ipred[j]/N_0 = \sum_k A[j,k]\alpha_k$, the six
independent entries are
\begin{align}
  F_{11} &= N_0\alpha_1^2\!\left(
    \tfrac{1}{\alpha_1} + \tfrac{1}{\alpha_1+\alpha_2}
    + \tfrac{1}{\alpha_1+\alpha_2+\alpha_3}\right), \label{eq:F11}\\
  F_{22} &= N_0\alpha_2^2\!\left(
    \tfrac{1}{\alpha_1+\alpha_2} + \tfrac{1}{\alpha_1+\alpha_2+\alpha_3}
    + \tfrac{1}{\alpha_2+\alpha_3}\right), \label{eq:F22}\\
  F_{12} &= N_0\alpha_1\alpha_2\!\left(
    \tfrac{1}{\alpha_1+\alpha_2} + \tfrac{1}{\alpha_1+\alpha_2+\alpha_3}
    \right), \label{eq:F12}\\
  F_{13} &= N_0\alpha_1\alpha_3 \cdot
    \tfrac{1}{\alpha_1+\alpha_2+\alpha_3}, \label{eq:F13}
\end{align}
with $F_{33} = F_{11}$ and $F_{23} = F_{12}$ under the exchange
$\alpha_1 \leftrightarrow \alpha_3$. The single term in $F_{13}$
reflects that paths~1 and~3 share only the fully mixed row~3. The
high-attenuation information loss is visible directly: when one path is
bright ($\alpha_{k'} \gg \alpha_k$), $P_j$ is dominated by
$\alpha_{k'}$ and $1/P_j$ is small, so the bright path \emph{dilutes}
the information row~$j$ carries about the dim path. This dilution is
irreversible --- the matrix-level statement of the domination
principle.

\subsection{Equal-attenuation factorization}
\label{subsec:factorization}

At the symmetric operating point $x_1 = x_2 = x_3 = x$ the predicted
intensities are $P_j = \alpha, 2\alpha, 3\alpha, 2\alpha, \alpha$, and
substituting into Eqs.~\eqref{eq:F11}--\eqref{eq:F13} factors a common
$N_0\alpha = N_0 e^{-x}$:
\begin{equation}
  \FCRB(x) = N_0 e^{-x}\, \mathbf{M},
  \qquad
  \mathbf{M} =
  \begin{pmatrix}
    \tfrac{11}{6} & \tfrac{5}{6} & \tfrac{1}{3} \\[2pt]
    \tfrac{5}{6} & \tfrac{4}{3} & \tfrac{5}{6} \\[2pt]
    \tfrac{1}{3} & \tfrac{5}{6} & \tfrac{11}{6}
  \end{pmatrix},
  \label{eq:FIM_factored}
\end{equation}
e.g.\ $M_{11} = 1 + \tfrac{1}{2} + \tfrac{1}{3} = \tfrac{11}{6}$ and
$M_{22} = \tfrac{1}{2} + \tfrac{1}{3} + \tfrac{1}{2} = \tfrac{4}{3}$,
the rest following by centrosymmetry. Crucially $\mathbf{M}$ is
\emph{independent of $x$}: hence
$\FCRB^{-1}(x) = (e^x/N_0)\mathbf{M}^{-1}$ and the inflation ratios
depend only on $\mathbf{M}$ --- that is, only on the structure of
$\Amat$ --- through $r_k = \sqrt{3\,[\mathbf{M}^{-1}]_{kk}}$
(Section~\ref{subsec:closed_form}), and are exactly constant for all $x \geq 0$ within this
family. Away from equality they vary (Fig.~\ref{fig:dosefloor}). Two
regimes are worth separating here. At equal attenuation the
\emph{conditioning} of $\FCRB$ is fixed --- it is that of $\mathbf{M}$,
with $\kappa(\mathbf{M}) = 6$ --- and only the overall \emph{magnitude}
changes, vanishing as $e^{-x}$ so that information about all three paths
shrinks together. Away from equality $\FCRB$ instead becomes genuinely
\emph{ill-conditioned}: bright-path dilution collapses the dark-path
eigendirection (the middle path, when it is the darkest), and this ---
not any property of $\Amat$, which stays well-conditioned throughout
(Section~\ref{subsec:SVD}) --- is the matrix-level origin of the
diverging middle-path penalty.

\subsection{Closed-form CRBs and inflation ratios}
\label{subsec:closed_form}

Inverting $\mathbf{M}$ by cofactors,
\begin{equation}
  \mathbf{M}^{-1} = \frac{1}{9}
  \begin{pmatrix} 7 & -5 & 1 \\ -5 & 13 & -5 \\ 1 & -5 & 7 \end{pmatrix},
  \qquad \det\mathbf{M} = \tfrac{9}{4},
  \label{eq:Minv}
\end{equation}
so the per-path CRB and its inflation ratio read directly off the
diagonal of $\mathbf{M}^{-1}$,
\[
  \sigCRB_k = \frac{e^{x/2}}{\sqrt{N_0}}\sqrt{[\mathbf{M}^{-1}]_{kk}},
  \qquad
  r_k = \frac{\sigCRB_k}{\sigfair} = \sqrt{3\,[\mathbf{M}^{-1}]_{kk}}.
\]
With $[\mathbf{M}^{-1}]_{11} = [\mathbf{M}^{-1}]_{33} = \tfrac{7}{9}$ and
$[\mathbf{M}^{-1}]_{22} = \tfrac{13}{9}$ this gives at equal attenuation
$\sigCRB_1 = \sigCRB_3 = \tfrac{\sqrt{7}}{3}\, e^{x/2}/\sqrt{N_0}$ and
$\sigCRB_2 = \tfrac{\sqrt{13}}{3}\, e^{x/2}/\sqrt{N_0}$, giving the
inflation ratios
\begin{equation}
  r_1 = r_3 = \sqrt{\tfrac{7}{3}} \approx 1.528,
  \qquad
  r_2 = \sqrt{\tfrac{13}{3}} \approx 2.082,
  \label{eq:ratios}
\end{equation}
both exactly constant for all $x \geq 0$. Equivalently, the
\emph{bundle efficiencies} $\eta_k = 1/r_k^2$ are
$\eta_1 = \eta_3 = 3/7 \approx 43\%$ and $\eta_2 = 3/13 \approx 23\%$:
of the nine latent Poisson counts, aggregation into five sums retains
only this fraction of the single-source Fisher information. The
middle-path penalty of $\sqrt{13/7} \approx 1.36$ over the endpoints
comes entirely from $x_2$ lacking a dedicated row --- every row
informing $x_2$ shares its Poisson noise with another path --- and is a
structural property of $\Amat$, not of the attenuation level.

These constant ratios hold only at the symmetric point; they are a
reference, not a worst case. As Fig.~\ref{fig:dosefloor} shows, away
from equality the endpoint ratio rises toward $\sqrt{3} \approx 1.73$
when one endpoint darkens, while the middle-path ratio diverges when
$x_2$ is the darkest path and, conversely, falls toward~$1$ when both
endpoints are opaque (sources~1 and~3 then emit nothing and rows~2--4
become three clean observations of $x_2$). The full
attenuation-dependent CRB follows from numerically inverting
$\FCRB(\xvec)$ in Eq.~\eqref{eq:FIM_general}; the asymptotic limits are
catalogued in the supplementary material. Because the efficiency vector
$(\eta_1,\ldots,\eta_K)$ derives from $\Amat$ alone, it serves as a
figure of merit for comparing alternative $M\times K$
multiplexed-measurement geometries at equal dose --- within TCT (the
design lever pursued in the companion paper) and in other
aggregated-photon modalities sharing the same forward structure. The
overlapping-beam transmission scans of emission tomography
\citep{yu2000} are one such modality: the efficiencies and
dose-inflation accounting developed here apply there with only $\Amat$
respecified, offering a statistical-limit perspective complementary to
the reconstruction algorithms used in that setting.

\begin{figure}[t]   
\centering
\begin{tikzpicture}
\begin{axis}[
  width=13cm, height=8.4cm, ymode=log, log basis y=10,
  xlabel={path attenuation of the darkening path,\ \ $x=\mu L$
          \ \ (remaining two paths held at $x=3$)},
  ylabel={penalty vs.\ single-source CT\ \ $\sqrt{\mathrm{DIF}}
          =\sigma_{\mathrm{CRB}}/\sigma_{\mathrm{fair}}$},
  xmin=0, xmax=9.3, ymin=0.85, ymax=46,
  xtick={0,1,2,3,4,5,6,7,8,9},
  ytick={1,2,5,10,20,40}, yticklabels={1,2,5,10,20,40},
  grid=both, major grid style={gray!14}, minor grid style={gray!6},
  tick label style={font=\footnotesize}, label style={font=\footnotesize},
  clip=false ]

  \addplot[black, densely dotted, thick, domain=0:9.3, forget plot]{1};
  \addplot[gray,  densely dotted, domain=0:9.3, forget plot]{1.7321};
  \node[font=\scriptsize, anchor=center] at (axis cs:5.7,1.27){single-source CT};
  \node[font=\scriptsize, gray!50!black, anchor=east] at (axis cs:9.08,1.98){$\sqrt3$};

  \draw[gray!70, dashed] (axis cs:3,0.85) -- (axis cs:3,46);
  \node[font=\scriptsize, gray!55!black, anchor=south, fill=white, inner sep=1pt]
        at (axis cs:3,43){balanced $x_1{=}x_2{=}x_3$};

  \addplot[red!75!black, very thick, dashed, mark=square*, mark size=1.0pt,
           mark options={solid}]
    coordinates {(0.00,1.0930)(0.25,1.1173)(0.50,1.1475)(0.75,1.1846)(1.00,1.2298)(1.25,1.2846)(1.50,1.3504)(1.75,1.4289)(2.00,1.5217)(2.25,1.6309)(2.50,1.7589)(2.75,1.9081)(3.00,2.0817)(3.25,2.2830)(3.50,2.5159)(3.75,2.7849)(4.00,3.0947)(4.25,3.4509)(4.50,3.8594)(4.75,4.3271)(5.00,4.8616)(5.25,5.4716)(5.50,6.1667)(5.75,6.9579)(6.00,7.8578)(6.25,8.8804)(6.50,10.0419)(6.75,11.3603)(7.00,12.8565)(7.25,14.5538)(7.50,16.4787)(7.75,18.6615)(8.00,21.1362)(8.25,23.9415)(8.50,27.1215)(8.75,30.7257)(9.00,34.8107)};
  \addplot[blue!70!black, very thick, mark=*, mark size=1.0pt]
    coordinates {(0.00,1.0868)(0.25,1.1076)(0.50,1.1323)(0.75,1.1611)(1.00,1.1942)(1.25,1.2312)(1.50,1.2717)(1.75,1.3147)(2.00,1.3593)(2.25,1.4041)(2.50,1.4479)(2.75,1.4893)(3.00,1.5275)(3.25,1.5619)(3.50,1.5920)(3.75,1.6180)(4.00,1.6399)(4.25,1.6581)(4.50,1.6731)(4.75,1.6853)(5.00,1.6951)(5.25,1.7029)(5.50,1.7091)(5.75,1.7141)(6.00,1.7180)(6.25,1.7210)(6.50,1.7234)(6.75,1.7253)(7.00,1.7268)(7.25,1.7280)(7.50,1.7289)(7.75,1.7296)(8.00,1.7301)(8.25,1.7305)(8.50,1.7309)(8.75,1.7311)(9.00,1.7313)};

  \addplot[red!75!black, only marks, mark=*, mark size=1.5pt, forget plot] coordinates {(3,2.0817)};
  \addplot[blue!70!black, only marks, mark=*, mark size=1.5pt, forget plot] coordinates {(3,1.5275)};
  \node[font=\scriptsize, red!75!black,  anchor=south east, inner sep=1pt]
        at (axis cs:2.85,2.32){$\sqrt{13/3}$};
  \node[font=\scriptsize, blue!70!black, anchor=north, inner sep=1pt]
        at (axis cs:3.05,1.30){$\sqrt{7/3}$};

  \draw[red!60!black, thin] (axis cs:3.0,6.0) -- (axis cs:4.2,3.32);
  \node[font=\scriptsize, red!75!black, align=center, anchor=south east]
        at (axis cs:3.05,6.2){middle path $x_2$\\(no dedicated reading)};
  \node[font=\scriptsize, blue!70!black, anchor=south]
        at (axis cs:5.9,1.82){endpoint path $x_1$ (dedicated reading)};
\end{axis}
\end{tikzpicture}
\caption{The fundamental dose floor of the TCT bundle---a Cram\'er--Rao
result, independent of any algorithm and of the photon budget $N_0$.
The per-path penalty $\sqrt{\mathrm{DIF}}=\sigma_{\mathrm{CRB}}/
\sigma_{\mathrm{fair}}$ is the standard-deviation cost relative to the
equal-dose single-source CT floor $\sigma_{\mathrm{fair}}=e^{x/2}/
\sqrt{3N_0}$ (penalty $=1$), shown as one bundle path darkens while the
other two are held at a representative $x=3$. At the balanced point
($x=3$, all three paths equal) the curves cross their
\emph{equal-attenuation} values, which are exactly constant at every
dose level: $\sqrt{7/3}\approx1.53$ for the endpoints and
$\sqrt{13/3}\approx2.08$ for the middle. Away from balance the
asymmetry is stark: the endpoint penalty saturates at $\sqrt3\approx
1.73$, capped by its dedicated reading $Y_1/Y_5$, whereas the middle
path---never measured in isolation---diverges without bound through
bright-path dilution, exceeding $30\times$ by $x=9$. This structural,
asymmetric cost is what the learned prior is designed to mitigate.}
\label{fig:dosefloor}
\end{figure}

\subsection{Two-level gap decomposition}
\label{subsec:decomposition}

For any estimator achieving standard deviation $\sigAlg_k$, the gap to
the floor factors multiplicatively,
\begin{equation}
  \frac{\sigAlg_k}{\sigfair}
    = \underbrace{\frac{\sigAlg_k}{\sigCRB_k}}
        _{\text{algorithmic inefficiency}}
      \cdot
      \underbrace{\frac{\sigCRB_k}{\sigfair}}
        _{\text{physics penalty } r_k \ge 1}.
  \label{eq:decomposition}
\end{equation}
The physics penalty $r_k$ is fixed by $\Amat$; the
algorithmic inefficiency is the legitimate target of better estimation.
For unbiased estimators it is bounded below by~$1$, attained
asymptotically by the maximum-likelihood estimator. For a biased
estimator that incorporates a prior $p(\xvec)$ --- including a neural
network trained on a joint anatomical prior --- it can fall below~$1$,
in which case the operative bound is the Bayesian CRB rather than
Eq.~\eqref{eq:ratios} (Section~\ref{sec:discussion}). This is the regime the patient-trained network enters at high
attenuation on a single anatomy --- a concentrated-prior effect we
interpret as interpolation within that anatomy rather than recovery of
aggregation-lost information (Section~\ref{sec:discussion}). The two
factors of Eq.~\eqref{eq:decomposition} structure the results that
follow.

\section{Classical inversion: the SNN1 baseline}
\label{sec:classical}
Two classical baselines frame the learned estimator: a linear inversion
that fails for a diagnostic reason, and SNN1, the structured
per-bundle estimator we adopt as the classical ceiling throughout.

\subsection{A legitimate linear inverse, and why it underperforms}
\label{subsec:SVD}

The forward model is exactly linear in the transmissions
$\alpha_k = e^{-x_k}$, since $\Imeas/N_0 \approx \Amat\,\boldsymbol{\alpha}$,
so one can invert for $\boldsymbol{\alpha}$ by least squares and recover
$\hat{x}_k = -\ln\hat{\alpha}_k$ --- a legitimate linear procedure, not
a linearization. It is nonetheless only a baseline, dominated by the count-domain
estimators below, and the reasons are statistical rather than numerical.
Conditioning is not the culprit: $\Amat$ is well conditioned
($\kappa \approx 3.19$), so regularization buys nothing. Nor is the
$-\ln$ step: the single-row log estimate $-\ln(\Imeas[j]/N_0)$ is the
(asymptotically) efficient single-source MLE and attains the single-ray
CRB (Section~\ref{sec:forward}), so on the dedicated endpoint rows the
linear inverse is already near-optimal. The shortfall is concentrated on
the middle path, through three statistical --- not algebraic --- defects.
First, ordinary least squares weights all five rows equally, whereas the
Poisson variance of each $\Imeas[j]/N_0$ is proportional to its own mean
and so differs across rows and with attenuation; an unweighted inverse is
therefore not the efficient (Poisson-weighted) estimator away from equal
attenuation. Second, $x_2$ is recovered as a small difference of brighter
endpoint contributions, so its error inflates precisely when $x_2$ is the
darkest path --- the same bright-path dilution the Fisher analysis
quantifies (Section~\ref{sec:crb}), not a numerical artifact. Third, in
the photon-starved regime the recovered transmissions $\hat{\alpha}_k$
are small and can fall to or below zero, where $-\ln\hat{\alpha}_k$ is
biased or undefined --- a positivity breakdown that no regularization of
the well-conditioned $\Amat$ repairs. A count-domain estimator that respects the Poisson
statistics and the positivity constraint is therefore required; we
retain the linear inverse only as a warm-start --- for the SNN1
iteration below and for the network's V2 variant (supplementary
material).

\subsection{SNN1: symmetric Newton via $N_0$ subtraction}
\label{subsec:SNN1}

SNN1 works directly in the photon-count domain and exploits the
bundle's structural asymmetry. It is warm-started from the SVD
transmission estimate of Section~\ref{subsec:SVD} --- which supplies the
initial middle path $x_2$, the one path with no dedicated row --- and
then refines that estimate through two phases per iteration.
\textbf{Phase~A} anchors the endpoints from their dedicated rows (1
and~5), taking a damped Newton step toward each single-source MLE,
\begin{equation}
  x_k^{A} = -\ln\!\Bigl[\mathrm{clip}\bigl(N_0 e^{-x_k}
    + f\,\Delta N_k,\; 1,\, \infty\bigr)/N_0\Bigr],
  \label{eq:SNN1_A}
\end{equation}
where $\Delta N_k = \Imeas[r(k)] - N_0 e^{-x_k}$ is the dedicated-row
residual, with the endpoint row map $r(1)=1$ and $r(3)=5$ (so $x_3$ is
anchored from row~5, not row~3), and $f \in (0,1]$ an annealing fraction.
Clipping at one photon imposes an effective maximum attenuation and
biases estimates in photon-starved bundles. \textbf{Phase~B} handles the three
mixed rows (2--4), apportioning each row residual among its active paths
in proportion to their predicted fractional intensities (Fisher-score
weighting) and updating $x_k$ as the intensity-weighted mean of the
per-row estimates, clipped to $[0, 9.5]$ --- a box constraint set just
above the data support ceiling $x_{\max}=9.2$, so it binds only on the
over-shooting middle-path estimates at the highest bins (the origin of
the constrained-estimator bias noted in Fig.~\ref{fig:snn1crb}). The fraction $f$ anneals
linearly from $0.80$ to $0.05$ over at most 50 iterations, with early
stopping at $\max_k |x_k^{(n)} - x_k^{(n-1)}| < 5\times10^{-6}$.

This structure lets SNN1 reach the per-bundle Fisher limit wherever the
geometry allows it. On i.i.d.\ data it brings the endpoints $x_1, x_3$
to within $1$--$2\%$ of their CRBs across bins~4--10 --- essentially the
optimum for an unbiased classical estimator (Fig.~\ref{fig:snn1crb}).
The middle path is different: SNN1 sits at $1.03$--$1.05\times$ its CRB
at moderate attenuation, but the gap widens to $1.4$--$1.7\times$ at
bins~6--7, exactly the regime where the $\sqrt{13/3}$ penalty and
bright-path dilution starve the mixed rows of information about $x_2$.
This residual middle-path gap is not closed by SNN1: it reaches the
endpoint CRBs but stays above the middle-path bound, and whether the
shortfall is an intrinsic limit of per-bundle classical estimation or
SNN1-specific slack is not settled here (it would require the exact
constrained MLE). Either way, it is precisely what a learned
prior must address.

At the two highest bins (8--9) the comparison changes character. The
photon-starved middle-path estimates would overshoot $x_{\max}$, and the
$[0,9.5]$ box constraint clips them, trading bias for a reduced spread.
SNN1's standard deviation then dips \emph{below} the per-bundle CRB
(open markers in Fig.~\ref{fig:snn1crb}) --- but this is the ordinary
consequence of biasing an estimator, not a recovery of information: the
per-bundle CRB lower-bounds only the variance of \emph{unbiased}
estimators (Section~\ref{subsec:decomposition}), so these sub-CRB ratios
are not a fair efficiency comparison. The same escape from the unbiased
bound, achieved later by a learned prior rather than a hard clip, is
what the network exploits on patient data (Section~\ref{subsec:uc1}).

Convergence is fast on informative data (7--15
iterations on i.i.d.\ and patient bundles); the concentrated phantom
dataset provides weaker Phase~B gradients and needs up to 200
(supplementary material).

\begin{figure}[t]
\centering
\begin{tikzpicture}
\begin{axis}[
  width=12cm, height=7.4cm,
  xlabel={attenuation bin $k$\ \ ($\mu L\in[k{-}1,k)$)},
  ylabel={SNN1 StdDev $/$ per-bundle CRB},
  xmin=3.4, xmax=9.6, ymin=0.3, ymax=2.0,
  xtick={4,5,6,7,8,9}, ytick={0.5,1.0,1.5},
  grid=both, major grid style={gray!14}, minor grid style={gray!6},
  tick label style={font=\footnotesize}, label style={font=\footnotesize},
  clip=false ]

  \addplot[black, thick, dashed, domain=3.4:9.6, forget plot]{1};
  \node[font=\scriptsize, anchor=north] at (axis cs:5.15,0.965){per-bundle CRB (Fisher bound)};

  \addplot[blue!70!black, very thick, mark=*, mark size=1.4pt]
     coordinates {(4,1.051)(5,1.032)(6,1.020)(7,1.013)(8,1.018)(9,1.021)};

  \addplot[red!75!black, very thick, mark=square*, mark size=1.6pt]
     coordinates {(4,1.028)(5,1.047)(6,1.701)(7,1.395)};
  \addplot[red!75!black, very thick, dashed, forget plot]
     coordinates {(7,1.395)(8,0.741)(9,0.416)};
  \addplot[red!75!black, only marks, mark=square, mark size=2.1pt, forget plot]
     coordinates {(8,0.741)(9,0.416)};

  \node[font=\scriptsize, blue!70!black, anchor=south west] at (axis cs:4.05,1.24)
     {endpoints $x_1,x_3$: $\le 5\%$ above CRB};
  \draw[blue!70!black, thin] (axis cs:4.5,1.23) -- (axis cs:5.0,1.05);
  \node[font=\scriptsize, red!75!black, anchor=south] at (axis cs:6.0,1.80)
     {middle $x_2$: classical apportionment gap};
  \node[font=\scriptsize, red!55!black, anchor=east] at (axis cs:7.72,0.70)
     {box-constrained (biased)};
  \draw[->, red!55!black, thin] (axis cs:7.78,0.71) -- (axis cs:7.95,0.735);
\end{axis}
\end{tikzpicture}
\caption{How close does the best classical per-bundle estimator get to
the dose floor? Ratio of SNN1 per-path StdDev to the per-bundle
Cram\'er--Rao bound, by attenuation bin (RND test set, $10^5$ bundles).
\textbf{Endpoint paths $x_1,x_3$} (blue): SNN1 lands within $5\%$ of the
bound at every bin and tightens toward it as attenuation rises---the
classical estimator essentially exhausts the per-bundle Fisher
information for the dedicated paths. \textbf{Middle path $x_2$} (red):
near the bound at moderate attenuation (bins~4--5), but a classical
\emph{apportionment gap} opens once the $\sqrt{13/3}$ Fisher penalty
bites ($1.70\times$ at bin~6, $1.40\times$ at bin~7). At bins~8--9 the
box-constrained estimate becomes biased and its StdDev falls below the
unbiased CRB (open markers), so the ratio there is not a fair efficiency
comparison. The residual middle-path gap at bins~6--7 is precisely what
a learned joint prior can close.}
\label{fig:snn1crb}
\end{figure}

\section{Physics-motivated neural-network inversion}
\label{sec:nn}
\subsection{Inputs and physics features}
\label{subsec:variants}

Two input variants test whether explicit preprocessing helps. \textbf{V1}
takes only the five normalized intensities $\{I_j/N_0\}$ and
$\log_{10} N_0$ (6 inputs); \textbf{V2} adds the SVD warm-start estimate
$\hat{\xvec}_{\mathrm{SVD}}$ (9 inputs). Both are augmented with ten
physics-derived features motivated by the FIM analysis: the five
log-transformed measurements $y_j = -\ln(I_j/N_0 + \varepsilon)$ (exact
for the dedicated endpoint rows, a low-noise endpoint signal) and five
row-level Fisher weights
$\hat{W}_k^{(j)} = A[j,k]\,e^{-\hat{x}_k^{\mathrm{SVD}}} /
\sum_{k'} A[j,k']\,e^{-\hat{x}_{k'}^{\mathrm{SVD}}}$, which let the
network down-weight rows carrying negligible Fisher information about a
given path --- the same adaptive weighting Phase~B of SNN1 performs.
Effective input dimensions are 16 (V1) and 19 (V2).

\subsection{Architecture}
\label{subsec:architecture}

The estimator is a physics-motivated residual network (PMRN): a linear
projection to a 256-dimensional hidden space, four gated residual blocks
\begin{equation}
  \mathbf{h}' = \mathbf{h} + \mathbf{g}\odot
    f\!\bigl(\mathbf{W}_2\, f(\mathrm{LN}(\mathbf{W}_1\mathbf{h}))\bigr),
  \qquad \mathbf{g} = \sigma(\mathbf{W}_g\mathbf{h}),
  \label{eq:resblock}
\end{equation}
with $f = \mathrm{ELU}$ and layer normalization \citep{ba2016} across
widths $256\!\to\!256\!\to\!128\!\to\!64$, and a zero-initialized linear
head producing a correction $\hat{\Delta}\xvec$. For V2 the estimate is
residual on the warm-start,
\begin{equation}
  \hat{\xvec} = \mathrm{softplus}\bigl(
    \hat{\xvec}_{\mathrm{SVD}} + \hat{\Delta}\xvec\bigr),
  \label{eq:output}
\end{equation}
so at initialization $\hat{\Delta}\xvec = \mathbf{0}$ and the network
reproduces the SVD estimate --- a stable starting point --- while the
softplus enforces $\hat{x}_k > 0$. V1 applies the correction without the
SVD offset. Both variants have $\approx 8.1\times10^5$ parameters.

\subsection{Datasets: an informativeness spectrum}
\label{subsec:data}

We evaluate on three datasets that share the forward
model~\eqref{eq:forward} and the per-bundle TCM dose model defined
below, but differ in how much joint structure a prior
could exploit (Table~\ref{tab:datasets}). They form a deliberate
progression from no structure to a rich single-anatomy prior, which is
what lets us separate algorithmic skill from prior exploitation in
Section~\ref{sec:results}.

\begin{table}[!t]
\centering\footnotesize
\renewcommand{\arraystretch}{1.15}
\caption{The three evaluation datasets, ordered by how much joint
structure a prior could exploit. $\sigma_{\bar{x}}$ is the bundle-mean
standard deviation (the dominant structure indicator); $\rho(x_1,x_2)$
the intra-bundle adjacent-path correlation.}
\label{tab:datasets}
\begin{tabular}{lccc}
\toprule
Property & RND & SGS & PIS \\
\midrule
Line-integral source & i.i.d.\ Beta mixture & geometric chest phantom & single patient CT \\
Max $\mu L$ & 9.2 & 8.93 & 9.2 \\
Bundle-mean std $\sigma_{\bar{x}}$ & $\approx 2.50$ & $\approx 0.44$ & $\approx 0.78$ \\
Intra-bundle $\rho(x_1,x_2)$ & $\approx 0$ & $\approx 0$ & $\approx -0.25$ \\
Bundles (total) & $14\times10^6$ & $68.2\times10^6$ & $68.2\times10^6$ \\
Split strategy & random 80/10/10 & $z$-region (contiguous) & slice-level (permuted) \\
\bottomrule
\end{tabular}
\end{table}

\textbf{RND} draws $(x_1,x_2,x_3)$ i.i.d.\ (fixed seed~42) from a
three-component Beta mixture on $[0,9.2]$ with density
\[
  p(x)=0.4\,\mathrm{Be}(2,4)+0.3\,\mathrm{Be}(4,4)+0.3\,\mathrm{Be}(6,2).
\]
The independence is physically unrealistic but
gives clean algorithmic control: intra- and inter-bundle correlations
are identically zero, so there is no joint prior to learn and the
network can only match the per-bundle Fisher content.

\textbf{SGS} extracts bundles from an analytical chest phantom
(geometric lungs, heart, ribs, spine) with exact line integrals.
Adjacent paths are essentially uncorrelated; the prior information lives
almost entirely in a concentrated bundle-mean distribution
($\sigma_{\bar{x}}\approx0.44$ versus $\approx2.5$ for i.i.d.),
encoding smooth but geometrically simple anatomy. Contiguous $z$-region
splitting is essential: a randomized shuffle leaked $z$-correlated
bundles across splits and collapsed the train/val gap to $0.3\%$,
whereas the region split restores a healthy $9\%$.

\textbf{PIS} extracts bundles from TCT fan-beam projections of a real
patient CT volume ($256\times256\times200$). \emph{All PIS bundles come
from a single patient}. The $70/15/15$ split assigns whole slices to
train/val/test by a random permutation (seed~42); because adjacent
slices remain anatomically correlated, this removes the within-slice
leakage of a naive bundle shuffle but does \emph{not} hold out
contiguous $z$-regions as the SGS split does, so the PIS sub-floor
figures should be read as an \emph{optimistic} within-anatomy bound
rather than a validated generalization result. The PIS results
characterize the informativeness of \emph{one} patient's anatomical
prior --- any sub-floor performance is interpolation within that
anatomy --- with multi-patient generalization the subject of a companion
paper that acquires a patient corpus. PIS shows real intra-bundle
correlation ($\rho(x_1,x_2)\approx-0.25$) and an L-shaped joint
distribution (concentrated in air and lung, with a bone and vascular
tail); its bundle-mean std ($\approx0.78$) lies between RND and SGS, an
ordering that reflects geometric simplicity rather than realism (PIS is
clinically closest).

The effective air-scan count is drawn per bundle from a TCM model,
$N_{0,b} = \mathrm{clip}(K_b\, e^{\bar{x}_b},\, 7.5\times10^4,\,
3\times10^5)$ with $K_b\sim\mathrm{LogUniform}(1397,5586)$, and the empirical
$\sigfair$ used in all dose comparisons is computed from these
per-bundle $N_{0,b}$ on the test split. Within a bundle this single
$N_{0,b}$ is shared by all of its rays: the model modulates flux from
bundle to bundle but holds it uniform across the bundle's three line
integrals. Relaxing this uniform-within-bundle flux --- co-designing the
source-flux profile through bow-tie filtration and tube-current
modulation --- is a substantial dose lever pursued in the companion
paper.

\subsection{Loss and training}
\label{subsec:loss}

Training minimizes a three-term loss
$\mathcal{L} = \mathcal{L}_{\mathrm{Huber}}
 + 0.30\,\mathcal{L}_{\mathrm{log}}
 + 0.05\,\mathcal{L}_{\mathrm{Poisson}}$. The Huber term (on
$\hat{x}_k - x_k$) is the primary supervision, robust to photon-starved
outliers; the log-domain term equalizes relative error across the full
range; and the Poisson-deviance term --- the exact bundle negative
log-likelihood in intensity space --- keeps the gradient informative at
high attenuation, where the MSE gradient on $x_k$ vanishes with
$e^{-x_k}$. A five-epoch Huber-only warmup protects the zero-initialized
head before the deviance term is enabled. Networks are trained with
AdamW \citep{loshchilov2019} under cosine-annealed warm restarts
\citep{loshchilov2017}, batch size 2048, for 80 epochs (best validation
checkpoint near epoch~59), on a single GTX~1080; full hyperparameters
are in the supplementary material.

\section{Results}
\label{sec:results}
Results follow the RND~$\to$~SGS~$\to$~PIS progression of
Section~\ref{subsec:data}, which is also a progression in how much joint
prior the network can exploit. All values are pooled per-bin error
standard deviations on held-out test splits; full per-path tables are in
the supplementary material, and Figs~\ref{fig:nn}
and~\ref{fig:priorquality} carry the per-bin values.

\subsection{RND: the algorithmic baseline}
\label{subsec:uc0}

On i.i.d.\ data the network has no joint prior to learn, so it can only
realize the per-bundle Fisher content --- and it does. Against the
linear (SVD) inverse it lowers the pooled error by $2.2$--$3.5\times$ at
bins~6--9, where the unweighted SVD inverse degrades sharply under photon
starvation (Section~\ref{subsec:SVD}).
Against the equal-dose floor, NN/$\sigfair$ stays above the endpoint
penalty $r_1 = \sqrt{7/3} \approx 1.53$ at every bin (peaking near $3.9$
at bins~6--7, where the harder middle path dominates the pooled figure):
with $p(x_1,x_2,x_3)$ independent there is no prior to carry the
estimate below the per-path CRB, exactly as theory requires. The
substantive comparison is against SNN1. On the endpoints the two agree
to within a few percent (the NN marginally ahead, $2$--$6\%$, by
implicitly using $x_2$ to sharpen the endpoints). On the middle path
they diverge with attenuation: at bin~5, where $x_2$ is still well
measured, SNN1 is $14\%$ better --- the learned prior adds a small bias
where the data are informative, the only bin and path on which SNN1
wins --- but above bin~5, as the $\sqrt{13/3}$ penalty collapses the
middle-path Fisher information, the NN closes the widening SNN1 gap by
$31$--$47\%$ at bins~6--9 through its learned joint prior. A V1/V2
ablation is essentially a tie (best-epoch losses equal to within
$0.1\%$, per-bin ratios within $[0.99,1.04]$): the SVD warm-start
carries no information the network has not already extracted from the
raw measurements --- evidence that a sufficiently expressive network
learns the nonlinear inverse implicitly.

\subsection{SGS: a concentrated but unrealistic prior}
\label{subsec:uc3}

The phantom dataset adds a prior, but the wrong kind. NN/$\sigfair$
stays above~1 at every bin (Fig.~\ref{fig:priorquality}): the network
beats SNN1 at high attenuation (bins~6--9: $33$--$67\%$) but never
crosses the equal-dose floor. The reason is concentration without
anatomical realism --- the phantom's bundle-mean distribution is narrow
($\sigma_{\bar{x}}\approx0.44$) but encodes smooth geometric shapes
rather than the variable tissue structure of real anatomy, so the
learned prior is useful yet not genuinely informative below the
frequentist CRB. An architectural ablation sharpens the point: NN~V4,
which locks the endpoints to their single-source MLEs and learns only
$x_2$, \emph{beats} the full joint network V2 at bins~1--7 on SGS,
because the phantom prior supplies no useful endpoint information and
V2's learned endpoints are merely noisier than the analytic ones; the
supplementary material tabulates this V2/V4 reversal across both
datasets. The reversal is sharpest at low attenuation: the V2 curve in
Fig.~\ref{fig:priorquality} spikes to $6.9\times\sigfair$ at bin~2,
where V2's noisy phantom endpoints dominate and the small
low-attenuation $\sigfair$ magnifies the ratio, whereas the
analytic-endpoint V4 stays near $1.8\times$ at the same bin
(supplementary Table~S8). This is the exact reverse of the patient-data result below, and it shows that
whether the joint prior helps a given path depends entirely on what the
prior actually encodes.

\subsection{PIS: a single-anatomy prior, and interpolation within it}
\label{subsec:uc1}

On patient-derived bundles the picture changes sharply
(Fig.~\ref{fig:nn}). NN/$\sigfair$ reaches near-parity with the floor at
bin~6 ($1.028$), falls below it at bin~7 ($0.541$), and reaches $0.146$
at bin~9 (Table~\ref{tab:pis_main}). We are explicit about what this is and is not. All PIS bundles
come from one patient, so the network has learned that patient's joint
distribution $p(x_1,x_2,x_3)$, which at high attenuation is nearly
deterministic given the bundle mean; the sub-floor numbers are the
network \emph{interpolating} within that single anatomy's bundle
statistics --- a Bayesian-prior effect that places the estimate below
the frequentist CRB. It characterizes the informativeness of this
patient's prior, not a generalizable dose advantage, and it is a regime
simply inaccessible to any prior-free per-bundle estimator, whose error
cannot fall below $\sigfair \cdot r_1$. The middle-path behaviour makes
the mechanism visible: the NN's $x_2$ error is nearly flat across
bins~3--9 and actually \emph{decreases} above bin~5 --- the opposite of
the physics-dictated trend --- because the prior, not the measurement,
is doing the work; the per-path ratios confirm the crossing is
middle-path-only, with the endpoints held at $1.3$--$1.5\times\sigfair$
(Table~\ref{tab:pis_ratio_main}). The middle-path NN bias stays below
$6\times10^{-4}$ in magnitude across the sub-floor bins, so the RMSE
equals the StdDev to the precision shown: the crossing holds in RMSE as
well as in variance, and is not an artifact of shrinkage toward the
population mean. SNN1, by contrast, collapses on PIS: it hits its
iteration ceiling on every bundle and accumulates large positive bias
($+0.04$, $+0.26$, $+0.31$ at bins~6--8), a structural failure of
proportional apportionment on concentrated bundle distributions, so the
NN/SNN1 pooled ratio falls to $0.03$ at bin~7 and $0.009$ at bin~9. The
endpoints are consistent: the NN beats SNN1 from bin~3 onward
($15$--$25\%$), using the joint prior to carry endpoint information no
per-bundle estimator can access. Here the full joint network V2 wins at
every bin but~9, and the V4 ablation confirms why --- V4's $x_2$
estimate is unchanged whether its endpoints come from the single-source
MLE or from SNN1 (within $2\%$), so the learned anatomical prior, not
the endpoint inputs, drives middle-path accuracy almost entirely. V2 is
therefore the primary architecture.

\begin{table}[t]
\renewcommand{\arraystretch}{1.15}
\centering\small
\caption{PIS pooled per-bin StdDev ($1{,}022{,}976$ records; SNN1 on a
PIS-matched 100K reference). SNN1 hits its iteration ceiling on every
bundle (positive bias $+0.04,+0.26,+0.31$ at bins~6--8). \textbf{Bold}:
$\mathrm{NN}/\sigfair<1$, i.e.\ interpolation within one anatomy --- not
a dose claim. Bin~10 empty.}
\label{tab:pis_main}
\begin{tabular}{ccrrrrr}
\toprule
Bin & Range & SNN1 & $\sigfair$ & NN~V2 & NN~V4-B & NN/$\sigfair$ \\
\midrule
1 & $[0,1)$ & 0.00212 & 0.00176 & 0.00317 & 0.00261 & 1.796 \\
2 & $[1,2)$ & 0.00451 & 0.00326 & 0.00451 & 0.00536 & 1.383 \\
3 & $[2,3)$ & 0.00875 & 0.00551 & 0.00761 & 0.01147 & 1.382 \\
4 & $[3,4)$ & 0.01652 & 0.00917 & 0.01212 & 0.01723 & 1.322 \\
5 & $[4,5)$ & 0.05449 & 0.01643 & 0.01979 & 0.02369 & 1.205 \\
6 & $[5,6)$ & 0.36607 & 0.02665 & 0.02740 & 0.02750 & 1.028 \\
7 & $[6,7)$ & 0.82373 & 0.03954 & 0.02138 & 0.03703 & \textbf{0.541} \\
8 & $[7,8)$ & 1.00331 & 0.05408 & 0.01303 & 0.03685 & \textbf{0.241} \\
9 & $[8,9)$ & 1.10133 & 0.06937 & 0.01011 & 0.00864 & \textbf{0.146} \\
\bottomrule
\end{tabular}
\end{table}

\begin{table}[t]
\renewcommand{\arraystretch}{1.15}
\centering\footnotesize
\setlength{\tabcolsep}{4pt}
\caption{Per-path NN~V2/$\sigfair$ on the PIS test set (slice-level
split). Endpoints: average of $x_1,x_3$. \textbf{Bold}: middle-path
ratio $<1$. The sub-floor crossing is confined to the middle path; the
endpoints never cross. The middle-path NN bias is $\le 6\times10^{-4}$
in magnitude across the sub-floor bins, so RMSE equals StdDev to the
precision shown and the ratios are unchanged when formed from RMSE.}
\label{tab:pis_ratio_main}
\begin{tabular}{ccrrr}
\toprule
Bin & Range & Endpoints (avg) & Middle $x_2$ & $x_2$ bias \\
\midrule
1 & $[0,1)$ & 1.79 & --- & --- \\
2 & $[1,2)$ & 1.38 & --- & --- \\
3 & $[2,3)$ & 1.37 & 1.30 & $+0.0042$ \\
4 & $[3,4)$ & 1.34 & \textbf{0.96} & $+0.0003$ \\
5 & $[4,5)$ & 1.39 & \textbf{0.62} & $+0.0002$ \\
6 & $[5,6)$ & 1.47 & \textbf{0.41} & $+0.0006$ \\
7 & $[6,7)$ & 1.52 & \textbf{0.27} & $+0.0003$ \\
8 & $[7,8)$ & 1.48 & \textbf{0.18} & $-0.0001$ \\
9 & $[8,9)$ & --- & \textbf{0.15} & $-0.0006$ \\
\bottomrule
\end{tabular}
\end{table}

\begin{figure}[t]   
\centering
\begin{tikzpicture}
\begin{axis}[
  width=13cm, height=8.2cm, ymode=log, log basis y=10,
  xlabel={attenuation bin $k$\ \ ($\mu L\in[k{-}1,k)$)},
  ylabel={$\sqrt{\mathrm{DIF}}=\sigma/\sigma_{\mathrm{fair}}$ \ (pooled, achieved)},
  xmin=0.5, xmax=9.5, ymin=0.1, ymax=28,
  xtick={1,2,3,4,5,6,7,8,9}, ytick={0.2,0.5,1,2,5,10,20},
  yticklabels={0.2,0.5,1,2,5,10,20},
  grid=both, major grid style={gray!14}, minor grid style={gray!6},
  tick label style={font=\footnotesize}, label style={font=\footnotesize},
  clip=false ]

  \fill[green!55!black, opacity=0.08] (axis cs:0.5,0.1) rectangle (axis cs:9.5,1.0);
  \node[font=\scriptsize, green!40!black, anchor=west] at (axis cs:1.3,0.32)
     {NN beats equal-dose single-source CT};

  \addplot[black, thick, domain=0.5:9.5, forget plot]{1};
  \addplot[gray, densely dashed, domain=0.5:9.5, forget plot]{1.7321};
  \node[font=\scriptsize, anchor=south west] at (axis cs:6.55,1.06){single-source CT floor};
  \node[font=\scriptsize, gray!50!black, anchor=west] at (axis cs:5.4,2.05)
     {per-bundle CRB ($\sqrt3$, pooled)};

  \addplot[red!75!black, very thick, densely dashed, mark=triangle*, mark size=2.2pt]
     coordinates {(1,1.205)(2,1.383)(3,1.588)(4,1.802)(5,3.316)(6,13.736)(7,20.833)(8,18.552)(9,15.876)};
  \addplot[blue!70!black, very thick, mark=*, mark size=1.7pt]
     coordinates {(1,1.796)(2,1.383)(3,1.382)(4,1.322)(5,1.205)(6,1.028)(7,0.541)(8,0.241)(9,0.146)};

  \draw[red!70!black, thin] (axis cs:4.9,10.6) -- (axis cs:6,13.736);
  \node[font=\scriptsize, red!75!black, anchor=west] at (axis cs:1.0,11)
     {SNN1 (classical, ceiling-limited)};
  \node[font=\scriptsize, blue!70!black, anchor=north] at (axis cs:3.4,0.94)
     {NN V2 (learned prior)};
\end{axis}
\end{tikzpicture}
\caption{The learned-prior payoff on patient data (PIS, UC1; pooled
per-bin StdDev, ${\sim}10^6$ test records). The achieved dose inflation
$\sqrt{\mathrm{DIF}}=\sigma/\sigma_{\mathrm{fair}}$ is the same quantity
the Cram\'er--Rao analysis bounds from below, now \emph{achieved} by two
estimators. \textbf{SNN1} (red), the near-CRB classical estimator, runs
into its iteration ceiling on every PIS bundle and inflates to
$>20\times$ the single-source floor at high attenuation. \textbf{NN~V2}
(blue), trained on the single-patient sinogram distribution, instead
\emph{falls below} the equal-dose single-source CT floor from bin~7
onward ($0.541$ at bin~7, $0.146$ at bin~9): on this one anatomy the
concentrated joint prior carries information beyond the per-bundle
Fisher content, so the Bayesian estimator drops below both the
single-source floor and the frequentist per-bundle CRB ($\sqrt3$,
pooled). We read this sub-floor behaviour as interpolation within the
patient's bundle statistics --- a concentrated-prior effect, not a
generalizable dose advantage (the mismatched-prior failure in
Fig.~\ref{fig:priorquality} reinforces this). SNN1's high-attenuation
values are additionally ceiling-limited.}
\label{fig:nn}
\end{figure}
\subsection{Prior quality and the mismatch control}
\label{subsec:v2v4_uc5}

The three datasets resolve into a clean ordering on prior quality
(Fig.~\ref{fig:priorquality}): with no prior (RND) NN/$\sigfair$ holds
above the frequentist penalty; a concentrated but geometrically simple
prior (SGS) improves on the classical estimator yet cannot cross the
floor; only the realistic single-anatomy prior (PIS) crosses it, by
interpolation. The cross-evaluation closes the argument. A network
trained on SGS and applied to the PIS test set fails badly ---
$2$--$20\times$ SNN1 at bins~6--8, with large positive bias ($+0.6$ to
$+7.1$; Table~\ref{tab:mismatch_main}) --- because the concentrated phantom prior predicts
high-attenuation triples that do not occur in patient anatomy. That a
\emph{mismatched} concentrated prior is worse than none is the control
confirming that the PIS sub-floor result is prior-specific
interpolation rather than a recovered measurement capability; it also
indicates that for any future multi-patient estimator, prior
\emph{diversity} (spanning body habitus and pathology) will matter more
than prior concentration. Accordingly, a central focus of the companion
paper is a multi-patient corpus of more than 20 CT volumes assembled to
span that diversity.

\begin{table}[t]
\renewcommand{\arraystretch}{1.2}
\centering\footnotesize
\caption{Cross-dataset degradation on the PIS test set: pooled NN/SNN1
StdDev ratio at bins~6--8 for a matched (PIS-trained) versus a
mismatched (SGS-trained) prior. A mismatched concentrated prior is worse
than none.}
\label{tab:mismatch_main}
\begin{tabular}{ll}
\toprule
Training prior & Outcome on PIS (bins~6--8) \\
\midrule
PIS (matched)    & $0.04$--$0.33\times$ SNN1; prior works \\
SGS (mismatched) & $2$--$20\times$ SNN1; substantial OOD failure \\
\bottomrule
\end{tabular}
\end{table}

\begin{figure}[t]   
\centering
\begin{tikzpicture}
\begin{axis}[
  width=13cm, height=8.2cm, ymode=log, log basis y=10,
  xlabel={attenuation bin $k$\ \ ($\mu L\in[k{-}1,k)$)},
  ylabel={$\sqrt{\mathrm{DIF}}=\sigma/\sigma_{\mathrm{fair}}$ \ (NN~V2, achieved)},
  xmin=0.5, xmax=9.5, ymin=0.1, ymax=9,
  xtick={1,2,3,4,5,6,7,8,9}, ytick={0.2,0.5,1,2,5},
  yticklabels={0.2,0.5,1,2,5},
  grid=both, major grid style={gray!14}, minor grid style={gray!6},
  tick label style={font=\footnotesize}, label style={font=\footnotesize},
  clip=false ]

  \fill[green!55!black, opacity=0.08] (axis cs:0.5,0.1) rectangle (axis cs:9.5,1.0);
  \node[font=\scriptsize, green!40!black, anchor=west] at (axis cs:1.1,0.22)
     {below single-source CT floor};

  \addplot[black, thick, domain=0.5:9.5, forget plot]{1};
  \addplot[gray, densely dashed, domain=0.5:9.5, forget plot]{1.7321};
  \node[font=\scriptsize, gray!50!black, anchor=west] at (axis cs:5.97,1.42)
     {per-bundle CRB ($\sqrt3$)};
  \node[font=\scriptsize, anchor=north east] at (axis cs:9.4,0.97){single-source CT floor};

  \addplot[violet!70!black, very thick, mark=triangle*, mark size=2pt]
     coordinates {(1,2.64)(2,2.06)(3,1.95)(4,2.13)(5,2.92)(6,3.93)(7,3.80)(8,3.18)(9,2.64)};
  \addplot[orange!85!black, very thick, densely dashed, mark=square*, mark size=1.8pt]
     coordinates {(1,4.308)(2,6.862)(3,3.003)(4,2.545)(5,3.125)(6,4.516)(7,4.130)(8,3.225)(9,1.258)};
  \addplot[blue!70!black, very thick, mark=*, mark size=1.7pt]
     coordinates {(1,1.796)(2,1.383)(3,1.382)(4,1.322)(5,1.205)(6,1.028)(7,0.541)(8,0.241)(9,0.146)};

  \node[font=\scriptsize, orange!85!black, anchor=west] at (axis cs:2.15,6.9){SGS (phantom prior)};
  \draw[violet!70!black, thin] (axis cs:6.45,2.6) -- (axis cs:6.6,3.78);
  \node[font=\scriptsize, violet!70!black, anchor=west] at (axis cs:5.5,2.42){RND (no prior)};
  \draw[blue!70!black, thin] (axis cs:6.4,0.58) -- (axis cs:7.0,0.541);
  \node[font=\scriptsize, blue!70!black, anchor=east] at (axis cs:6.35,0.58){PIS (patient prior)};

  \node[font=\scriptsize, red!60!black, align=left, anchor=west] at (axis cs:5.0,5.9)
     {Prior mismatch (SGS$\to$PIS),\\ inflates to $11\times$ and beyond (off scale $\uparrow$)};
\end{axis}
\end{tikzpicture}
\caption{Prior quality sets the achievable dose floor. Achieved
$\sqrt{\mathrm{DIF}}=\sigma/\sigma_{\mathrm{fair}}$ for the same NN~V2
architecture trained and tested on each of the three datasets, ordered
by sinogram structure. \textbf{RND} (no prior; i.i.d.\ line integrals):
stays well above the single-source floor at every bin---without a joint
prior the per-bundle estimator cannot beat equal-dose single-source CT.
\textbf{SGS} (concentrated but geometrically simple phantom prior):
improves at high attenuation yet still never crosses below $1$---a
phantom prior is not anatomical information. \textbf{PIS} (single patient
anatomy): a concentrated prior rich enough to drive
$\sqrt{\mathrm{DIF}}$ below the floor at high attenuation --- but, being
single-patient, this reflects interpolation within that one anatomy
rather than a generalizable dose gain. The cross-evaluation
(SGS-trained network applied to PIS) instead inflates to $11$--$391\times$
with large positive bias, showing the effect is prior-specific: a
\emph{mismatched} prior is worse than none. The lesson for any future
deployment is that prior \emph{diversity} (many patients), not prior
concentration, is what a genuine dose claim would require.}
\label{fig:priorquality}
\end{figure}

\section{Inter-bundle structure and the companion problem}
\label{sec:interbundle}
Paper~1 inverts each bundle in isolation, and in that setting the
per-bundle CRBs of Section~\ref{sec:crb} are the correct benchmark. But
the bundles are not independent. Because neighbouring detector readings
sample almost the same anatomy, the line integrals of adjacent bundles
are strongly correlated, and that correlation persists over surprisingly
large separations.

The same-path correlation of the middle path is $\rho \approx 0.999$
for immediately adjacent bundles and decays only slowly with
separation, in three geometrically distinct directions
(Fig.~\ref{fig:interbundle}): channel-to-channel within a projection
view, row-to-row across the multi-row detector ($\rho$ still
$\approx 0.71$ at a 127-row separation), and view-to-view between
adjacent gantry angles. A second, subtler coupling also appears: the
intra-bundle endpoint anti-correlation $\rho(x_1,x_3)\approx-0.5$ in
patient data propagates to adjacent bundles rather than arising from
any geometric path overlap.

This redundancy is exactly what a per-bundle estimator cannot use, and
it is the natural route past the single-anatomy interpolation limit of
Section~\ref{sec:results}: an estimator that couples bundles along all
three directions would draw on a far larger and more diverse pool of
constraints than any single bundle provides, and the appropriate lower
bound then becomes the Bayesian CRB incorporating the sinogram prior
rather than the per-bundle CRB.

\begin{figure}[t]
\centering
\resizebox{\columnwidth}{!}{%
\begin{tikzpicture}
\begin{axis}[
  width=12cm, height=7.4cm, xmode=log, log basis x=10,
  xlabel={bundle separation $s$\ \ (channels / detector rows / views)},
  ylabel={same-path correlation $\rho(x_2)$ \ (PIS)},
  xmin=0.9, xmax=150, ymin=0.6, ymax=1.06,
  xtick={1,2,5,10,30,64,127}, xticklabels={1,2,5,10,30,64,127},
  ytick={0.6,0.7,0.8,0.9,1.0},
  grid=both, major grid style={gray!14}, minor grid style={gray!6},
  tick label style={font=\footnotesize}, label style={font=\footnotesize},
  clip=false ]

  \addplot[blue!70!black, very thick, mark=*, mark size=1.7pt]
     coordinates {(1,0.9995)(10,0.988)(30,0.950)(64,0.874)(100,0.763)(127,0.708)};
  \addplot[orange!85!black, very thick, densely dashed, mark=square*, mark size=1.9pt]
     coordinates {(1,0.9995)(30,0.716)};

  \node[font=\scriptsize, anchor=west] at (axis cs:1.25,1.035)
     {all three directions coincide at $s{=}1$: $\rho\approx0.9995$};

  \node[font=\scriptsize, blue!70!black, anchor=west] at (axis cs:44,0.95)
     {row-to-row ($z$)};
  \node[font=\scriptsize, orange!85!black, anchor=west] at (axis cs:32,0.70)
     {channel-to-channel};

  \node[font=\scriptsize, align=left, anchor=west] at (axis cs:1.1,0.69)
     {\textbf{Paper 1} inverts each bundle independently.\\
      \textbf{Paper 2} couples neighbouring bundles\\
      along all three directions.};
\end{axis}
\end{tikzpicture}}
\caption{The inter-bundle structure Paper~1 leaves unexploited (PIS
patient sinogram). Paper~1 inverts each $5\times3$ bundle in isolation,
yet neighbouring bundles are strongly correlated. Shown is the same-path
correlation of the middle path $x_2$ as a function of bundle
separation~$s$. Immediately adjacent bundles ($s{=}1$)---in the channel
direction (within a projection view), the row direction (across detector
rows, $z$), and the view direction (across gantry angles)---are all
$\rho\approx0.9995$. The decay with $s$ is plotted for the channel and
row directions; the view direction was characterised only at adjacent
views ($\rho\approx0.9995$), so it shares the $s{=}1$ value but is not
drawn as a separate decay curve. The correlation falls only slowly: the
row ($z$) signal stays above $0.7$ even at a 127-row separation, and the
faster-decaying channel direction still holds $\rho\approx0.72$ at 30
channels. This near-rank-deficient redundancy across the sinogram is the
basis of the companion paper's sinogram-domain constraint, which couples
bundles along all three directions rather than inverting each in
isolation.}
\label{fig:interbundle}
\end{figure}

These observations define the companion program, which attacks the dose
floor of Eq.~\eqref{eq:dif} on complementary fronts: generalizing the
learned prior to a multi-patient corpus (so the sub-floor behaviour
reported here can be tested for genuine generalization), coupling
neighbouring bundles through the inter-bundle correlations above,
optimizing bundle geometry through the Fisher efficiencies $\eta_k$, and
co-designing the source-flux profile. Paper~1 fixes the per-bundle floor
against which that program is measured.

\section{Discussion}
\label{sec:discussion}
\subsection{What the learned prior does --- and what it does not}
\label{subsec:what_nn_does}

The per-path comparison with SNN1 gives a clean account of the
network's role, worth stating precisely because the headline sub-floor
result is easy to over-read. On the endpoints, SNN1 already saturates
the per-path bundle CRB on i.i.d.\ data ($1.01$--$1.05\times$), and the
network improves on it only marginally there ($2$--$6\%$) by using the
mixed rows to tie the endpoints to $x_2$. On the middle path the
behaviour changes sharply at bin~5: below it, where $x_2$ is still well
measured, SNN1 sits near its CRB and the learned prior is a mild
liability ($14\%$ worse at bin~5 on RND); above it, the $\sqrt{13/3}$
penalty collapses the middle-path Fisher information and the prior takes
over, closing the SNN1 gap by $31$--$47\%$. The transition is the whole
story: a prior helps exactly where the measurement stops carrying
information, and not before.

On patient data this crossover goes further, and here we are careful
about its meaning. When a prior is incorporated, the achievable variance
can fall below the frequentist CRB; the formal lower bound is then the
Bayesian (van Trees) CRB \citep{trees1968},
\begin{equation}
  \sigma^{\mathrm{Bayes}}_k \;=\;
    \sqrt{\Bigl[\bigl(\mathbb{E}_{\xvec}[\mathbf{F}(\xvec)]
      + \mathbf{J}_{\mathrm{prior}}\bigr)^{-1}\Bigr]_{kk}}
    \;\le\; \sigCRB_k,
  \label{eq:BCRB}
\end{equation}
a \emph{lower} bound on the achievable per-component RMS error: it is the
diagonal of the matrix van Trees inequality
$\mathbb{E}[(\hat{\xvec}-\xvec)(\hat{\xvec}-\xvec)^{\top}]
\succeq (\mathbb{E}_{\xvec}[\mathbf{F}(\xvec)] + \mathbf{J}_{\mathrm{prior}})^{-1}$,
in which the prior Fisher information
$\mathbf{J}_{\mathrm{prior}} = -\mathbb{E}[\nabla^2\ln p(\xvec)]$ adds to
the measurement's. At high attenuation $F_{kk}$ is small, so the
prior term dominates and the bound can drop below $\sigfair$. The PIS
result --- middle-path error flat and even \emph{decreasing} with
attenuation, the opposite of the physics-dictated trend --- is the
visible signature of this regime: the patient's joint distribution is
nearly deterministic given the bundle mean, so the network reads $x_2$
off the prior rather than the starved measurement. Two cautions follow.
First, this is a property of the prior, not of TCT --- and not of the
neural architecture: a concentrated enough prior lowers the Bayesian
bound for any estimation problem, any estimator that encodes that prior
(standard Bayesian or MAP estimation among them \citep{stayman2000}) would attain the same
sub-floor effect, and the same network on i.i.d.\ data (no prior) never
crosses the floor. The network is one convenient way to encode
$p(\xvec)$, not the source of the gain.
Second, our single-patient evaluation does not establish that the
network \emph{attains} the Bayesian CRB --- only that its sub-floor
behaviour is consistent with this regime. The result quantifies how
informative one anatomy's prior is; it is not a dose claim. SNN1's
complementary collapse on PIS (limit-cycling, bias $+0.26$ at bin~7) is
the same point from the other side: a classical estimator with no prior
cannot enter this regime at all.

\subsection{Prior quality is the binding constraint}
\label{subsec:spectrum}

The three datasets isolate prior \emph{quality} as the variable that
matters (Fig.~\ref{fig:priorquality}). The improvement from RND to SGS
to PIS is not architectural --- it is the same network --- but a
progression in how much real structure the prior encodes: i.i.d.\ data
offers none, the phantom offers concentration without anatomy, and the
patient offers both; only the last crosses the floor. The V4/V2
reversal across SGS and PIS confirms the diagnosis: with an
uninformative endpoint prior (SGS) the analytically locked architecture
wins, whereas with an informative one (PIS) the full joint network does.
The cross-evaluation makes prior quality a safety issue as much as a
performance one: a mismatched prior does not merely underperform,
it confidently overestimates (bias up to $+7.1$), because it expects
high-attenuation triples the test anatomy does not contain. Such
failures are in principle detectable at inference --- the ratio of
measurement to prior Fisher information,
$F_{kk}/[\mathbf{J}_{\mathrm{prior}}]_{kk}$, flags bundles where the
prior has overridden the data --- and down-weighting in that regime is a
natural safeguard for a deployed estimator. The practical implication is
the one already stated: prior \emph{diversity}, supplied by a
multi-patient corpus, is the prerequisite for any generalizable claim.

\subsection{Implicit physics learning}
\label{subsec:convergence_result}

The V1/V2 ablation is small but clarifying: raw measurements alone (V1)
and measurements plus the SVD warm-start (V2) converge to
indistinguishable accuracy ($<0.1\%$). A sufficiently expressive network
learns the complete nonlinear inverse from the five raw counts and
$N_0$, so the explicit preprocessing can be dropped without loss --- the
``physics-motivated'' features accelerate training but are not required
at inference.

\subsection{Limitations and scope}
\label{subsec:limitations}

Several boundaries delimit these results. The analysis is monochromatic
and scatter-free; the two-level loss decomposition is independent of
spectral effects, so the qualitative structure should persist
polychromatically, but beam hardening and scatter are not modeled.
Polychromaticity also opens an opportunity unavailable to single-source
CT: simultaneous sources with \emph{distinct} spectral signatures
(different kVp, filtration, or photon-counting energy bins) would add
source-identifying information that partially counteracts the
exponential-summation loss --- a direction we pursue separately. The
most binding limitation is prior quality and scope: the sub-floor PIS
result is single-patient and characterizes one anatomy's prior, and a
properly regularized classical estimator would be a fairer
high-attenuation comparator than SNN1, whose iteration ceiling is itself
a structural artifact on concentrated distributions. All reported
metrics also use a single training seed (42): the large test sets make
the per-bin estimates statistically tight, but the network's
training-seed variability is not separately characterized. Finally, we defer
image-domain evaluation: propagating per-bundle errors through
reconstruction now would assess an incomplete system, since the
inter-bundle constraints and multi-patient priors of the companion
paper are expected to reduce middle-path error substantially. The
appropriate test --- prior-informed, flux-optimized TCT against
equal-dose single-source CT in the image domain --- is a milestone
downstream of that work.

\section{Conclusion}
\label{sec:conclusion}
We have studied the per-bundle inversion problem of Temporal CT ---
recovering three line-integral attenuations from five summed Poisson
intensities in the $5\times3$ geometry --- and established both its
statistical limits and the performance of practical estimators. Because
the forward model sums exponentials, performance separates into two
levels: a structural aggregation loss fixed by the incidence matrix,
and a reducible algorithmic inefficiency that better estimators can
close. The closed-form CRBs make the first level exact --- at equal
attenuation the Fisher information factorizes and the inflation ratios
are the constants $\sqrt{7/3}\approx1.53$ (endpoints) and
$\sqrt{13/3}\approx2.08$ (middle path), with bundle efficiencies
$\eta_k$ that serve as a figure of merit for any $M\times K$ multiplexed
geometry. Against this floor, SNN1 exhausts the per-bundle Fisher
information at the endpoints (within $1$--$2\%$ of CRB) but leaves the
characteristic $1.4$--$1.7\times$ middle-path gap, and a
physics-motivated residual network closes much of that gap through a
learned joint prior --- while the V1/V2 convergence shows the five raw
counts are the complete per-bundle observation, so the network learns
the nonlinear inverse directly from them without engineered features.

The network's behaviour across the RND~$\to$~SGS~$\to$~PIS datasets
traces a single axis: prior quality. With no prior (RND) it cannot cross
the equal-dose floor; with a concentrated but geometrically simple prior
(SGS) it improves on the classical estimator yet still cannot cross it;
only with a realistic single-anatomy prior (PIS) does the middle-path
noise fall below the floor at high attenuation. We interpret that
sub-floor result as interpolation within one patient's bundle
statistics --- a Bayesian-prior effect that places the estimate below
the frequentist CRB, consistent with the van Trees
bound~\eqref{eq:BCRB} --- and not as a generalizable dose advantage; the
failure of a mismatched prior under cross-distribution testing confirms
that the effect is prior-specific. Establishing whether it survives
across patients is precisely the task we defer.

Neither the structure analysed nor the methodology is specific to TCT: a
Poisson sum of exponentials, in which $\log$ and $\sum$ do not commute,
is the signature of any acquisition that aggregates attenuated photons
before counting, and the FIM-efficiency and dose-inflation accounting
transfer to such problems with only the incidence matrix respecified.
For TCT, these per-bundle limits are the baseline against which a
companion program will be measured, attacking the dose floor on
complementary fronts: generalizing the prior with a multi-patient
corpus, coupling neighbouring bundles through their strong inter-bundle
correlations, optimizing geometry through the efficiencies $\eta_k$, and
co-designing the source-flux profile through bow-tie filtration and
tube-current modulation. The dose-inflation factor is the figure of
merit they share. What this paper establishes is the floor itself: a
closed-form, geometry-determined account of the statistical cost of
multiplexed photon aggregation, against which any such acquisition must
be measured.

\section*{Acknowledgments}
During the preparation of this work the author used Claude (Anthropic,
Claude Opus~4 series) for \LaTeX{} drafting and editing, manuscript
restructuring, copy-editing, reference formatting, and a critical read of
the text for internal consistency, and for engineering support in the
computational work, including guided selection of neural-network training
hyperparameters against validation loss (the settings reported in the
supplementary material are standard optimizer defaults with a
conservative epoch budget). Generative AI was not used to create, alter,
or interpret any research data, figures, or quantitative results: all
simulations, the patient-derived sinogram data, the estimators, and every
reported value are the author's own, produced by author-written code. The
author reviewed and verified all AI-assisted output, made all scientific
decisions, and takes full responsibility for the content, accuracy,
integrity, and originality of the work.

\section*{Data availability statement}
The datasets generated and analysed in this study are available
upon request. The synthetic random-line-integral (RND)
data can additionally be regenerated directly from the generative model
and parameters specified in Section~\ref{subsec:data}. The
patient-image-derived (PIS) data are derived from a clinical CT volume
and are shared upon request, subject to the privacy and
institutional constraints governing patient imaging data. The PIS volume was used in fully de-identified form, contains no patient-identifiable information, and enters the analysis only as anonymized attenuation line integrals. Analysis code
(the SNN1 classical estimator and the physics-motivated residual
network) is available upon request.


\clearpage

\setcounter{section}{0}
\setcounter{subsection}{0}
\setcounter{table}{0}
\setcounter{figure}{0}
\setcounter{equation}{0}
\renewcommand{\thesection}{S\arabic{section}}
\renewcommand{\thetable}{S\arabic{table}}
\renewcommand{\theequation}{S\arabic{equation}}

\section*{Supplementary Material}
\noindent\textbf{Per-Bundle Statistical Limits and Learned-Prior
Inversion in Multiplexed X-ray Imaging, with Application to Temporal
CT}\\[0.4em]
\noindent\textit{This supplement provides the material referenced as ``supplementary material'' in the main paper: the Fisher-information limits beyond equal attenuation (S1), classical-baseline diagnostics (S2), full network training detail (S3), the complete per-path and per-bin result tables (S4), and the cross-dataset prior-mismatch analysis (S5). Section, equation, and table numbers in the main paper are referred to descriptively; all cross-references within this supplement are prefixed with ``S''. Datasets are denoted RND (i.i.d.\ synthetic), SGS (analytical phantom), and PIS (single-patient), as in the main text.}

\bigskip
\suppressfloats[t]
\section{Fisher information beyond equal attenuation}
\label{sec:fim}
The main paper derives the bundle Fisher information matrix
$\mathbf{F}$, its six independent entries $F_{11},\dots,F_{23}$, and the
equal-attenuation factorization $\mathbf{F}(x)=N_0e^{-x}\mathbf{M}$ that
yields the constant inflation ratios $r_1=r_3=\sqrt{7/3}\approx1.528$
and $r_2=\sqrt{13/3}\approx2.082$. Two extensions are collected here:
the behaviour of the ratios away from equal attenuation, and a per-bin
numerical confirmation of the two-level gap decomposition.

\subsection{Asymptotic inflation-ratio regimes}
Away from the equal-attenuation family the ratios
$r_k=\sigCRB_k/\sigfair(x_k)$ vary, and the limiting regimes are
informative (Table~\ref{tab:inflation_limits}). Both limits follow from
the Schur complement
$[\mathbf{F}^{-1}]_{11}=\bigl(F_{11}-[F_{12},F_{13}]\,\mathbf{C}^{-1}
[F_{12},F_{13}]^{\top}\bigr)^{-1}$, with $\mathbf{C}$ the lower-right
$2\times2$ block of $\mathbf{F}$.

\begin{table}[t]
\renewcommand{\arraystretch}{1.2}
\centering\small
\caption{Fair-comparison inflation ratios
$r_k=\sigCRB_k/\sigfair(x_k)$ across key attenuation regimes. The
endpoint ratio saturates at its dedicated-row value $\sqrt3$ when a
single endpoint darkens; both endpoints darkening drives $x_2$ toward
single-source efficiency ($r_2\to1$).}
\label{tab:inflation_limits}
\begin{tabular}{p{4.6cm}ccc}
\toprule
Regime & $r_1$ & $r_2$ & $r_3$ \\
\midrule
Equal attenuation (any $x$)            & $\sqrt{7/3}$ & $\sqrt{13/3}$ & $\sqrt{7/3}$ \\
$x_1\!\to\!\infty$, others fixed       & $\to\sqrt{3}$ & $\downarrow$ & $\sqrt{7/3}$ \\
$x_1,x_3\!\to\!\infty$, $x_2$ fixed    & $\to\infty$ & $\to 1$ & $\to\infty$ \\
\bottomrule
\end{tabular}
\end{table}

When $x_1\to\infty$ with $x_2,x_3$ fixed, the mixed rows are dominated
by the bright paths and carry negligible information about $x_1$; only
the dedicated single-source row~1 contributes, so
$\sigCRB_1\to\sigma^{\mathrm{single}}(x_1)$ and $r_1\to\sqrt{3}$. When
both endpoints darken ($\alpha_1,\alpha_3\to0$), sources~1 and~3 emit no
photons and rows~2,3,4 become independent Poisson observations of
$x_2$, so the middle-path efficiency $\eta_2\to1$ and $r_2\to1$ --- the
aggregation penalty vanishes precisely when the competing paths stop
contributing photons.

\subsection{Per-bin two-level gap decomposition}
Table~\ref{tab:decomposition} evaluates the main paper's two-level
decomposition at equal-attenuation bin centres for an analytical
reference flux $N_0=122{,}000$. The endpoint ratio $r_1=\sqrt{7/3}$ is
constant to five figures across all bins, confirming numerically that
the physics penalty is attenuation-independent; the remaining column
(Blk~3) is the analytical standard deviation an estimator at the bound
would carry once the residual algorithmic block is included.

\begin{table}[t]
\renewcommand{\arraystretch}{1.15}
\centering\footnotesize
\setlength{\tabcolsep}{4pt}
\caption{Two-level gap decomposition at equal-attenuation bin centres
(analytical reference, $N_0=122{,}000$). $r_1=\sqrt{7/3}\approx1.528$ is
exactly constant, confirming attenuation-independence of the physics
penalty.}
\label{tab:decomposition}
\begin{tabular}{ccccccc}
\toprule
Bin & Range & $x_c$ & $\sigfair$ & $\sigCRB_1$ & $r_1$ & Blk~3 \\
\midrule
1 & $[0,1)$ & 0.50 & 0.00365 & 0.00558 & 1.528 & 0.00631 \\
2 & $[1,2)$ & 1.50 & 0.00602 & 0.00919 & 1.528 & 0.01042 \\
3 & $[2,3)$ & 2.50 & 0.00992 & 0.01516 & 1.528 & 0.01718 \\
4 & $[3,4)$ & 3.50 & 0.01636 & 0.02499 & 1.528 & 0.02832 \\
5 & $[4,5)$ & 4.50 & 0.02697 & 0.04122 & 1.528 & 0.04669 \\
6 & $[5,6)$ & 5.50 & 0.04448 & 0.06795 & 1.528 & 0.07700 \\
7 & $[6,7)$ & 6.50 & 0.07334 & 0.11206 & 1.528 & 0.12695 \\
8 & $[7,8)$ & 7.50 & 0.12093 & 0.18480 & 1.528 & 0.20938 \\
9 & $[8,9)$ & 8.50 & 0.19940 & 0.30471 & 1.528 & 0.34539 \\
\bottomrule
\end{tabular}
\end{table}

\section{Classical baselines: SVD diagnostics and SNN1}
\label{sec:classical-supp}
The main paper establishes why the linear (SVD) inverse is only a
suboptimal baseline --- it is near-optimal on the dedicated endpoint
rows, but an unweighted least-squares inverse is not Poisson-efficient,
the middle path inherits the structural bright-path dilution, and the
$-\ln$ of a photon-starved, possibly non-positive transmission is biased
at high attenuation, none of which any regularization of the
well-conditioned $\Amat$ repairs. We do not repeat that diagnostic here;
instead we tabulate the SNN1 classical baseline against which the learned
network is measured.

\subsection{SNN1 on RND}
Table~\ref{tab:snn1_rnd} reports SNN1 per-path StdDev on $10^5$ i.i.d.\
test bundles. On the endpoints SNN1 sits within $1$--$2\%$ of the bundle
CRB across bins~4--9, essentially the optimum for an unbiased classical
estimator. On the middle path it tracks the CRB at moderate attenuation
($1.03\times$ at bin~4) but opens a $1.4$--$1.7\times$ gap at bins~6--7
under the $\sqrt{13/3}$ penalty; at bins~8--9 the box constraint biases
the estimate below the unconstrained CRB, so the ratio there is no
longer a fair efficiency measure (marked${}^{*}$).

\begin{table}[t]
\renewcommand{\arraystretch}{1.15}
\centering\footnotesize
\setlength{\tabcolsep}{4pt}
\caption{SNN1 per-path StdDev on $10^5$ RND test bundles.
${}^{*}$SNN1 falls below the unconstrained CRB at bins~8--9 for $x_2$
through box-constraint regularization (biased).}
\label{tab:snn1_rnd}
\begin{tabular}{ccrrcrr}
\toprule
Bin & Range & $\sigCRB_{\mathrm{ep}}$ & SNN1 $x_1/x_3$ & (ratio)
    & $\sigCRB_{x_2}$ & SNN1 $x_2$ \\
\midrule
4 & $[3,4)$ & 0.0136 & 0.0143 & (1.05) & 0.0214 & 0.0220 \\
5 & $[4,5)$ & 0.0217 & 0.0224 & (1.03) & 0.0487 & 0.0510 \\
6 & $[5,6)$ & 0.0344 & 0.0351 & (1.02) & 0.1171 & 0.1992 \\
7 & $[6,7)$ & 0.0540 & 0.0547 & (1.01) & 0.2958 & 0.4127 \\
8 & $[7,8)$ & 0.0848 & 0.0863 & (1.02) & 0.7298 & $0.5410^{*}$ \\
9 & $[8,9)$ & 0.1244 & 0.1270 & (1.02) & 1.4384 & $0.5981^{*}$ \\
\bottomrule
\end{tabular}
\end{table}

\subsection{SNN1 on SGS}
On the analytical phantom SNN1 requires $N_{\mathrm{iter}}=200$ to
converge (Table~\ref{tab:snn1_sgs}; mean iterations $150$--$199$ across
bins), because the concentrated bundle-mean distribution
($\sigma_{\bar{x}}\approx0.44$) supplies little Phase-B gradient. This is
the converged reference used in the SGS comparisons; on PIS, by
contrast, SNN1 limit-cycles at its iteration ceiling on every bundle
(main paper).

\begin{table}[t]
\renewcommand{\arraystretch}{1.15}
\centering\footnotesize
\setlength{\tabcolsep}{4pt}
\caption{SGS SNN1 baseline: pooled per-path StdDev on the full
$6.8\times10^6$-bundle test set ($N_{\mathrm{iter}}=200$). Bias
$<0.013$ at bins~1--8. Bin~8 pooled exceeds bin~9 (constrained-estimator
effect); bin~10 empty.}
\label{tab:snn1_sgs}
\begin{tabular}{ccrrr}
\toprule
Bin & Range & $N$ paths & SNN1 (pooled) & Mean iter \\
\midrule
1 & $[0,1)$ & 3{,}624{,}893 & 0.00141 & 166.6 \\
2 & $[1,2)$ &   592{,}428   & 0.00452 & 194.1 \\
3 & $[2,3)$ &   901{,}332   & 0.00753 & 198.7 \\
4 & $[3,4)$ & 1{,}654{,}313 & 0.01488 & 199.3 \\
5 & $[4,5)$ & 3{,}017{,}283 & 0.03006 & 196.1 \\
6 & $[5,6)$ & 3{,}959{,}288 & 0.08912 & 179.0 \\
7 & $[6,7)$ & 4{,}590{,}485 & 0.21809 & 175.2 \\
8 & $[7,8)$ & 2{,}049{,}507 & 0.46434 & 179.7 \\
9 & $[8,9)$ &    69{,}991   & 0.37703 & 151.4 \\
\bottomrule
\end{tabular}
\end{table}

\section{Network: input variants and training detail}
\label{sec:training}
\subsection{Input variants}
Both variants share the five normalized intensities $I_j/N_0$ and
$\log_{10}N_0$ as the raw measurement block, and the ten FIM-motivated
features of the main paper; they differ only in whether the SVD
warm-start is supplied (Table~\ref{tab:variants}).

\begin{table}[t]
\renewcommand{\arraystretch}{1.15}
\centering\small
\caption{Network input variants. All variants share the raw measurement
block ($n_{\mathrm{raw}}=6$) and the ten physics features; effective
input dimensions are 16 (V1) and 19 (V2).}
\label{tab:variants}
\begin{tabular}{clcl}
\toprule
Variant & Additional inputs & $n_{\mathrm{in}}$ & Role \\
\midrule
V1 & none & 16 & raw measurements alone \\
V2 & $\hat{\xvec}_{\mathrm{SVD}}$ & 19 & SVD warm-start \\
\bottomrule
\end{tabular}
\end{table}

\subsection{Loss function}
Training minimizes a three-term loss,
\begin{equation}
\mathcal{L} = w_1\mathcal{L}_{\mathrm{Huber}}
            + w_2\mathcal{L}_{\mathrm{log}}
            + w_3\mathcal{L}_{\mathrm{Poisson}},
\quad (w_1,w_2,w_3)=(1.0,\,0.30,\,0.05).
\label{eq:loss}
\end{equation}
The \emph{Huber} term (transition $\delta=1.0$) on
$\hat{x}_k-x_k^{\mathrm{true}}$ is the primary supervision, preferred
over plain MSE so that photon-starved outliers do not dominate the
gradient. The \emph{Poisson deviance}
\begin{equation}
\mathcal{L}_{\mathrm{Poisson}}=\frac{1}{5}\sum_{j=1}^{5}
  \Bigl[\hat{N}_j-\Imeas[j]+\Imeas[j]\ln\tfrac{\Imeas[j]}{\hat{N}_j}\Bigr],
\quad \hat{N}_j=\Ipred[j](\hat{\xvec}),
\label{eq:poisson}
\end{equation}
is the exact Poisson negative log-likelihood up to a constant; its
gradient acts in intensity space and stays informative at high
attenuation, where the MSE gradient on $x_k$ vanishes with $e^{-x_k}$.
The \emph{log-domain} term
\begin{equation}
\mathcal{L}_{\mathrm{log}}=\frac{1}{3}\sum_{k=1}^{3}
  \bigl(\ln(\hat{x}_k+1)-\ln(x_k^{\mathrm{true}}+1)\bigr)^2
\label{eq:logloss}
\end{equation}
equalizes relative error across the full $[0,9.2]$ range. A five-epoch
warmup uses $\mathcal{L}_{\mathrm{Huber}}$ alone before the deviance term
is enabled, protecting the zero-initialized output head from the large
early-epoch deviance.

\subsection{Optimizer, schedule, and hardware}
Networks are trained with AdamW \citep{loshchilov2019}
($\eta=3\times10^{-4}$, weight decay $10^{-5}$, gradient clipping at
global norm $1.0$) under cosine annealing with warm restarts (SGDR
\citep{loshchilov2017}; $T_0=20$, $T_{\mathrm{mult}}=2$, oscillating
between $\eta=3\times10^{-4}$ and $\eta_{\min}=3\times10^{-6}$). Batch
size is $2048$ and training runs for $80$ epochs, with the best
validation-loss checkpoint saved per variant (typically near epoch~59;
a mild overfitting trend of $\lesssim1.6\%$ validation-loss increase
appears in the final 20 epochs). All experiments use a single NVIDIA
GTX~1080 (8\,GB) with PyTorch 2.3.1, at roughly one hour per epoch.
These settings were obtained by guided manual tuning against validation
loss rather than a systematic grid or random search; the learning rate
and weight decay are standard AdamW defaults, and the 80-epoch budget is
deliberately conservative, as validation loss reaches within
$\sim\!0.5\%$ of its final value by epoch~20 (the first warm restart).

\section{Full per-path and per-bin results}
\label{sec:results-supp}
This section collects the complete per-bin and per-path numbers that the
main paper's figures summarize. Throughout, $\sqrt{\mathrm{DIF}}=
\sigma/\sigfair$ is computed from the per-bundle $N_0$ on each test
split. As in the main paper, any sub-floor value on PIS
($\mathrm{NN}/\sigfair<1$) reflects interpolation within a single
anatomy --- a concentrated-prior (Bayesian) effect --- and not a
generalizable dose advantage.

\subsection{Numerical confirmation of the inflation ratios}
Table~\ref{tab:crb_equal} evaluates the bundle CRBs at equal attenuation
across the full range at $N_0=10^5$; the inflation ratios
$r_1=\sqrt{7/3}$ and $r_2=\sqrt{13/3}$ are constant to the figures shown,
the empirical counterpart of the factorization in S1.

\begin{table}[t]
\renewcommand{\arraystretch}{1.15}
\centering\footnotesize
\setlength{\tabcolsep}{4pt}
\caption{Bundle CRBs at equal attenuation $x_1=x_2=x_3=x$, $N_0=10^5$.
$N$ is the transmitted count $N_0e^{-x}$. The inflation ratios are
exactly constant across the full attenuation range.}
\label{tab:crb_equal}
\begin{tabular}{rrrrrrr}
\toprule
$x$ & $N$ & $\sigfair$ & $\sigCRB_1$ & $\sigCRB_2$ & $r_1$ & $r_2$ \\
\midrule
0.0 & 100{,}000 & 0.00183 & 0.00279 & 0.00380 & 1.528 & 2.082 \\
2.0 &  13{,}534 & 0.00497 & 0.00759 & 0.01034 & 1.528 & 2.082 \\
4.0 &   1{,}832 & 0.01350 & 0.02062 & 0.02809 & 1.528 & 2.082 \\
6.0 &      248  & 0.03669 & 0.05603 & 0.07635 & 1.528 & 2.082 \\
8.0 &       34  & 0.09977 & 0.15232 & 0.20759 & 1.528 & 2.082 \\
9.2 &       10  & 0.17326 & 0.26472 & 0.36074 & 1.528 & 2.082 \\
\bottomrule
\end{tabular}
\end{table}

\subsection{Pooled per-bin StdDev}
Tables~\ref{tab:rnd_pooled} and~\ref{tab:sgs_pooled} give the pooled
per-bin StdDev behind Figs.~4--5 of the main paper for RND and SGS; the
corresponding PIS pooled table is in the main-paper Results.

\begin{table}[t]
\renewcommand{\arraystretch}{1.15}
\centering\footnotesize
\setlength{\tabcolsep}{3.5pt}
\caption{RND pooled per-bin StdDev ($1.4\times10^6$ records). NN/$\sigfair$
stays above the endpoint penalty $r_1\approx1.53$ at every bin: with no
joint prior the estimator is bounded by the per-path CRB.}
\label{tab:rnd_pooled}
\begin{tabular}{ccrrrrr}
\toprule
Bin & Range & $\sigCRB_1$ & SVD & $\sigfair$ & NN~V2 & NN/$\sigfair$ \\
\midrule
1 & $[0,1)$ & --- & 0.00458 & 0.00190 & 0.00501 & 2.64 \\
2 & $[1,2)$ & --- & 0.00642 & 0.00297 & 0.00611 & 2.06 \\
3 & $[2,3)$ & --- & 0.01144 & 0.00485 & 0.00946 & 1.95 \\
4 & $[3,4)$ & --- & 0.02486 & 0.00796 & 0.01695 & 2.13 \\
5 & $[4,5)$ & --- & 0.06140 & 0.01306 & 0.03813 & 2.92 \\
6 & $[5,6)$ & 0.02811 & 0.18867 & 0.02149 & 0.08435 & 3.93 \\
7 & $[6,7)$ & 0.04626 & 0.44521 & 0.03537 & 0.13443 & 3.80 \\
8 & $[7,8)$ & 0.07563 & 0.64202 & 0.05812 & 0.18501 & 3.18 \\
9 & $[8,9)$ & 0.11094 & 0.70231 & 0.09258 & 0.24473 & 2.64 \\
\bottomrule
\end{tabular}
\end{table}

\begin{table}[t]
\renewcommand{\arraystretch}{1.15}
\centering\footnotesize
\setlength{\tabcolsep}{3pt}
\caption{SGS pooled per-bin StdDev ($6.8\times10^6$-bundle test set).
NN/$\sigfair>1$ at every bin: the concentrated phantom prior improves on
SNN1 at high attenuation but never crosses the floor. V4 (locked
endpoints) wins bins~1--7; V2 wins bins~8--9. Bin~10 empty.}
\label{tab:sgs_pooled}
\begin{tabular}{ccrrrrr}
\toprule
Bin & Range & SNN1 & $\sigfair$ & NN~V2 & V2/$\sigfair$ & V4/$\sigfair$ \\
\midrule
1 & $[0,1)$ & 0.00141 & 0.00159 & 0.00683 & 4.308 & 1.241 \\
2 & $[1,2)$ & 0.00452 & 0.00338 & 0.02318 & 6.862 & 1.752 \\
3 & $[2,3)$ & 0.00753 & 0.00575 & 0.01726 & 3.003 & 1.746 \\
4 & $[3,4)$ & 0.01488 & 0.00937 & 0.02385 & 2.545 & 1.986 \\
5 & $[4,5)$ & 0.03006 & 0.01464 & 0.04576 & 3.125 & 2.778 \\
6 & $[5,6)$ & 0.08912 & 0.02481 & 0.11206 & 4.516 & 4.297 \\
7 & $[6,7)$ & 0.21809 & 0.03943 & 0.16281 & 4.130 & 4.068 \\
8 & $[7,8)$ & 0.46434 & 0.06174 & 0.19907 & 3.225 & 3.392 \\
9 & $[8,9)$ & 0.37703 & 0.09671 & 0.12168 & 1.258 & 1.340 \\
\bottomrule
\end{tabular}
\end{table}

\subsection{Per-path detail on PIS}
The per-path breakdown (Table~\ref{tab:pis_perpath}) makes the mechanism
explicit: the middle-path NN StdDev is flat and even \emph{decreasing}
with attenuation ($0.013\to0.010$ across bins~3--9), the opposite of the
physics-dictated trend, while SNN1 diverges as its apportionment
collapses. Expressed against the floor (per-path NN/$\sigfair$ table in
the main-paper Results), the
sub-floor behaviour is confined to the middle path --- it crosses below
$\sigfair$ at bin~4 and reaches $0.15$ at bin~9 --- whereas the
endpoints, already near their CRB classically, stay at
$1.34$--$1.79\times\sigfair$ and never cross.

\begin{table}[t]
\renewcommand{\arraystretch}{1.15}
\centering\footnotesize
\setlength{\tabcolsep}{4pt}
\caption{PIS per-path StdDev, NN~V2 vs.\ PIS-matched SNN1 (100K
bundles). The middle path ($x_2$) is flat/decreasing for the network and
divergent for SNN1.}
\label{tab:pis_perpath}
\begin{tabular}{cc rr rr rr}
\toprule
 & & \multicolumn{2}{c}{$x_1$} & \multicolumn{2}{c}{$x_2$} & \multicolumn{2}{c}{$x_3$} \\
\cmidrule(lr){3-4}\cmidrule(lr){5-6}\cmidrule(lr){7-8}
Bin & Range & NN & SNN1 & NN & SNN1 & NN & SNN1 \\
\midrule
3 & $[2,3)$ & 0.0074 & 0.0089 & 0.0131 & 0.0173 & 0.0076 & 0.0085 \\
4 & $[3,4)$ & 0.0118 & 0.0150 & 0.0130 & 0.0294 & 0.0117 & 0.0148 \\
5 & $[4,5)$ & 0.0213 & 0.0280 & 0.0123 & 0.1006 & 0.0214 & 0.0286 \\
6 & $[5,6)$ & 0.0373 & 0.0448 & 0.0115 & 0.5090 & 0.0369 & 0.0452 \\
7 & $[6,7)$ & 0.0572 & 0.0667 & 0.0106 & 0.8668 & 0.0577 & 0.0674 \\
8 & $[7,8)$ & 0.0760 & 0.0892 & 0.0096 & 1.009 & 0.0910 & 0.1031 \\
9 & $[8,9)$ & ---    & ---    & 0.0101 & 1.101 & ---    & ---    \\
\bottomrule
\end{tabular}
\end{table}

\section{Prior mismatch: cross-dataset degradation}
\label{sec:mismatch}
\subsection{V2 versus V4 by prior richness}
The competition between the full joint network (V2) and the
endpoint-locked ablation (V4) resolves as a function of prior richness.
On PIS (rich anatomical prior) V2 wins at every bin except~9: the joint
prior benefits the endpoints and the middle path simultaneously. On SGS
(concentrated but geometrically simple) V4 wins bins~1--7: the phantom
prior encodes no useful endpoint information, so V4's analytically
locked endpoints beat V2's noisier learned ones (V4 bias $\leq0.018$
throughout; V2 shows $-0.094$ at bin~8). On both datasets V4's $x_2$
estimate is nearly independent of endpoint quality (the two endpoint
variants of V4 agree to within $2\%$), confirming that middle-path
accuracy is driven by the learned prior rather than the endpoint inputs.

\subsection{Cross-dataset prior mismatch}
To quantify the risk of an inappropriate prior, the SGS-trained network
is evaluated on the PIS test set. The result is a substantial
out-of-distribution failure: pooled $\mathrm{NN}/\sigfair$ reaches
$11$--$391$ across bins, $2$--$20\times$ the SNN1 baseline at
bins~6--8, with large positive bias ($+0.6$ to $+7.1$; see the cross-dataset table
in the main-paper Results). The network has learned
``$x_2\approx\mathrm{high}$'' from the concentrated phantom distribution
and applies it to PIS bundles passing through air and lung, where it is
badly wrong.

This is the control behind the moderated reading of the PIS result in
the main paper. A concentrated prior drives the middle-path variance
below the floor only when it \emph{matches} the test anatomy; a
mismatched concentrated prior is worse than no prior at all. The
sub-floor PIS behaviour is therefore prior-specific interpolation, and
prior \emph{diversity} --- a corpus spanning many patients, body
habitus, and pathology --- not prior concentration, is the prerequisite
for any generalizable claim.

\section{Efficiency across geometries: a second incidence matrix}
\label{sec:second_matrix}

The per-bundle efficiency vector depends on the incidence matrix alone,
which makes it a quantity that can be evaluated for a candidate geometry
\emph{before} any data are acquired. At equal attenuation the Fisher
information factorizes as $\mathbf{F}=N_0 e^{-x}\,\mathbf{M}$ with
\begin{equation}
  \mathbf{M} = \mathbf{A}^{\top}\operatorname{diag}(1/n_j)\,\mathbf{A},
  \qquad n_j=\textstyle\sum_k A_{jk},
  \label{eq:M_general}
\end{equation}
where $n_j$ is the number of paths contributing to reading $j$; for the
$5\times3$ matrix of the main text this reduces to the $\mathbf{M}$
inverted there. The fraction of single-source Fisher information
retained by path $k$, referenced to a single-source scan of the same
total budget $N_S N_0$, and the corresponding inflation ratio are
\begin{equation}
  \eta_k = \frac{1}{N_S\,[\mathbf{M}^{-1}]_{kk}},
  \qquad
  r_k = \sqrt{N_S\,[\mathbf{M}^{-1}]_{kk}} = \eta_k^{-1/2},
  \label{eq:eta_general}
\end{equation}
both independent of $x$ and of $N_0$.

\paragraph{A second geometry.}
We instantiate \eqref{eq:eta_general} for the sliding-window bundle of
$N_S$ simultaneously-overlapping sources --- the natural generalization
of the $5\times3$ problem ($N_S=3$). For $N_S$ sources and $K=N_S$ paths
the incidence matrix is the $(2N_S{-}1)\times N_S$ staircase in which
path $k$ occupies $N_S$ consecutive readings, the central reading
summing all $N_S$ sources. The $N_S=4$ case is a genuinely different
$7\times4$ geometry,
\begin{equation}
  \mathbf{A}_{4} =
  \begin{pmatrix}
  1&0&0&0\\ 1&1&0&0\\ 1&1&1&0\\ 1&1&1&1\\
  0&1&1&1\\ 0&0&1&1\\ 0&0&0&1
  \end{pmatrix},
  \label{eq:A4}
\end{equation}
for which \eqref{eq:eta_general} yields endpoint efficiency
$\eta=0.308$ ($r=1.80$) and, for each of the two interior paths,
$\eta=0.138$ ($r=2.69$) --- distinctly worse than the $43\%/23\%$ of the
three-source bundle. The efficiency vector therefore discriminates
between geometries, as a design figure of merit should.

\emph{Assumption.} We model the realized $N_S$-source bundle as this
clean staircase. The actual overlap pattern of the double-drum geometry
is fixed by the angular-sampling condition
$r_\omega\le(\Delta\theta_a/2\pi)\,N_S$ of \citet{besson2015spie} and may
differ in detail; the staircase is taken as the representative motif,
consistent with the model-problem stance of the main text.

\begin{table}[t]
\centering
\caption{Per-bundle efficiency across the $N_S$-source sliding-window
family at equal attenuation, against the acquisition-speed gain
$r_\omega=N_S/3$ of the same geometry \citep[Table~I]{besson2015spie}.
``Endpoint'' is a dedicated-reading path; ``central'' is the most
interior (worst) path. $N_S=3$ is the $5\times3$ model problem; $N_S=2$
recovers the two-source conjugate pair.}
\label{tab:eta_family}
\begin{tabular}{cccccc}
\toprule
$N_S$ & $r_\omega$ & $\eta_{\mathrm{end}}$ & $r_{\mathrm{end}}$
      & $\eta_{\mathrm{cen}}$ & $r_{\mathrm{cen}}$ \\
\midrule
2 & ---  & 0.667 & 1.225 & ---   & ---   \\
3 & 1.00 & 0.429 & 1.528 & 0.231 & 2.082 \\
4 & 1.33 & 0.308 & 1.803 & 0.138 & 2.693 \\
5 & 1.67 & 0.238 & 2.049 & 0.082 & 3.493 \\
6 & 2.00 & 0.194 & 2.273 & 0.058 & 4.143 \\
\bottomrule
\end{tabular}
\end{table}

\paragraph{Two regularities.}
The endpoint efficiency obeys the closed form
\begin{equation}
  \eta_{\mathrm{end}}(N_S) = \frac{N_S}{N_S^{2}-N_S+1},
  \qquad
  r_{\mathrm{end}}(N_S) = \sqrt{\frac{N_S^{2}-N_S+1}{N_S}},
  \label{eq:eta_end}
\end{equation}
verified for $N_S\le 8$: at $N_S=3$ it is the $3/7$ and $\sqrt{7/3}$ of
the main text, and at $N_S=2$ it reproduces the $\sqrt{3/2}\approx1.22$
of the conjugate-pair geometry. The interior paths degrade far faster
than the endpoints --- the central-path efficiency falls from $23\%$
($N_S=3$) through $14\%$, $8\%$, $6\%$ as $N_S=4,5,6$, with $r$ rising
from $2.08$ to $4.14$ (Table~\ref{tab:eta_family}).

\paragraph{The design trade-off.}
Table~\ref{tab:eta_family} sets these efficiencies against the
acquisition-speed gain $r_\omega=N_S/3$ that the same geometry buys.
Speed and efficiency are conjugate: doubling the speed gain from $1.0$
to $2.0$ ($N_S:3\to6$) drives the central-path inflation from $2.08$ to
$4.14$. The efficiency vector quantifies the dose side of that exchange,
so the number of simultaneous sources is a quantity to be chosen
\emph{against} these limits rather than fixed in advance --- the role the
per-bundle analysis is meant to play in geometry selection.

\paragraph{Bundle length.}
Holding $N_S=3$ fixed and lengthening the bundle ($K>3$) leaves the
endpoints near $35$--$43\%$ but starves deep-interior paths further (the
central-path $\eta$ falling below $5\%$ by $K=15$), confirming that the
$5\times3$ block already captures the endpoint/middle asymmetry and that its
middle-path figure is, if anything, optimistic for the interior of a
long sliding window.


\end{document}